\documentclass[12pt,a4paper,twoside]{amsart} 
\usepackage[latin1]{inputenc}
\usepackage{latexsym}
\usepackage{amsmath}
\usepackage{amssymb}
\usepackage{amsfonts}
\usepackage{amstext}
\usepackage{mathrsfs}               % \mathscr
\usepackage{bbm}                    % \mathbbm 
\usepackage{bbold}                  % \mathbb 
\usepackage{textcomp}               
\usepackage{enumerate}
\usepackage{units}
\usepackage{multicol}
\newcommand{\CC}{\mathbb{C}}

\newcommand{\NN}{\mathbb{N}}

\newcommand{\RR}{\mathbb{R}}

\newcommand{\ZZ}{\mathbb{Z}}
\newcommand{\supp}{\mathrm{supp}}

\newcommand{\Ran}{\mathrm{Ran}}
\newcommand{\Tr}{\mathrm{Tr}}
\newcommand{\Ker}{\mathrm{Ker}} 

\newcommand{\loc}{\mathrm{loc}}
\newcommand{\el}{\mathrm{el}}

\newcommand{\id}{\mathbbm{1}}% Identity         
\newcommand{\ve}{\varepsilon}% alternativ small greak letters
\newcommand{\vp}{\varphi}

\newcommand{\vk}{\varkappa}
\newcommand{\vr}{\varrho}

\newcommand{\const}{\mathfrak{c}}              %  constant
\newcommand{\wt}[1]{\widetilde{#1}}
\newcommand{\mr}[1]{\mathring{#1}}
\newcommand{\wh}[1]{\widehat{#1}}
\newcommand{\ol}[1]{\overline{#1}} 

\newcommand{\nf}[2]{\nicefrac{#1}{#2}}      % nice fractions

\newcommand{\slim}[1]{\underset{#1}{\textrm{s-lim}}}
\newcommand{\twlim}[1]{{\textrm{w-lim}_{#1}}}
\newcommand{\tslim}[1]{{\textrm{s-lim}_{#1}}}
\newcommand{\essinf}[1]{\underset{#1}{\mathrm{ess}\,\mathrm{inf}}}

\newcommand{\tessinf}[1]{{\mathrm{ess}\,\mathrm{inf}_{#1}}}
\newcommand{\tesssup}[1]{{\mathrm{ess}\,\mathrm{sup}_{#1}}}
\newcommand{\V}[1]{\mathbf{#1}}
\newcommand{\vsigma}{\boldsymbol{\sigma}}

\newcommand{\valpha}{\boldsymbol{\alpha}}
\newcommand{\vbeta}{\boldsymbol{\beta}}

\newcommand{\vmu}{\boldsymbol{\mu}}
\newcommand{\vnu}{\boldsymbol{\nu}}
\newcommand{\veps}{\boldsymbol{\varepsilon}}

\newcommand{\cA}{\mathcal{A}}
\newcommand{\cB}{\mathcal{B}} 

\newcommand{\cD}{\mathcal{D}} 

\newcommand{\cF}{\mathcal{F}}
\newcommand{\cS}{\mathcal{S}}

\newcommand{\cM}{\mathcal{M}}       
 
\newcommand{\sB}{\mathscr{B}} 
\newcommand{\sC}{\mathscr{C}}
\newcommand{\sD}{\mathscr{D}} 

\newcommand{\sF}{\mathscr{F}}

\renewcommand{\le}{\leqslant} % greater/less or equal         
\renewcommand{\ge}{\geqslant}
\renewcommand{\imath}{i}
\renewcommand{\Im}{\mathrm{Im}\,}
\renewcommand{\Re}{\mathrm{Re}\,}
\newcommand{\Spec}{\sigma}
\newcommand{\spe}{\sigma_{\mathrm{ess}}}
\newcommand{\Hf}{H_{\mathrm{f}}}
\newcommand{\ee}{\mathfrak{e}}
\newcommand{\pmax}{\mathfrak{p}}
\newcommand{\emax}{\epsilon_0}
\newcommand{\dom}{\cD}
\newcommand{\HR}{\mathscr{H}}

\newtheorem{theorem}{Theorem}[section]
\newtheorem{lemma}[theorem]{Lemma}
\newtheorem{proposition}[theorem]{Proposition}
\newtheorem{corollary}[theorem]{Corollary}
\newtheorem{definition}[theorem]{Definition}
\theoremstyle{remark}

\newtheorem{hypothesis}[theorem]{Hypothesis}
\newtheorem{example}[theorem]{Example}
\numberwithin{equation}{section}
\newcommand{\C}{\mathbb{C}} % Complex numbers
\newcommand{\Z}{\mathbb{Z}} % Integer numbers
\newcommand{\R}{\mathbb{R}} % Real numbers
\newcommand{\pf}{\V{p}_{\mathrm{f}}}

\renewcommand{\v}[1]{\mathbf{v}(#1)}

\newcommand{\p}{ \V{p}}
\newcommand{\x}{ \V{x}}

\newcommand{\q}{ \V{q}}
\newcommand{\G}{ \V{G}}
\renewcommand{\k}{ \V{k}}

\newcommand{\m}{\textrm{m}}
\newcommand{\Fock}{\sF}
\newcommand{\one}{\mathbbm{1}}
\newcommand{\form}{\mathcal{Q}}
\newcommand{\core}{\mathscr{D}}

\newcommand{\HamF}[1]{{H}_{#1}}
\newcommand{\Ham}[1]{\mathbbm{H}_{#1}}

\newcommand{\nr}{\mathrm{nr}}
\newcommand{\sr}{\mathrm{sr}}
\newcommand{\Bir}[3]{{K}_{#1}^{#3}(#2)}
\newcommand{\That}[1]{{T}_{#1}(\hat{\p})}
\newcommand{\Shat}{{S}(\hat{\p})}
\newcommand{\E}[1]{ {E}_{#1}}
\newcommand{\Egs}[1]{ \mathbbm{E}_{#1}}
\newcommand{\ww}{\mathbbm{w}}
\newcommand{\SP}[2]{\big\langle #1,#2 \big\rangle} %S. Product BIG 
\newcommand{\sps}[2]{\langle #1,#2 \rangle} %Scalar Product
\newcommand{\AAA}{\mathbbm{A}}
\title[Enhanced Binding]{On enhanced Binding and related effects in the 
non- and semi-relativistic Pauli-Fierz models}
\author{Martin K\"onenberg}
\author{Oliver Matte}
\begin{document}
\begin{abstract}
We prove enhanced binding and increase of
binding energies in the non- and semi-relativistic Pauli-Fierz models,
for arbitrary values of the fine-structure constant and the
ultra-violet cut-off, and
discuss the resulting improvement of exponential
localization of ground state eigenvectors.
For the semi-relativistic model we also discuss
the increase of the renormalized electron mass and
determine the linear leading order term in the asymptotics
of the self-energy, as the ultra-violet cut-off goes to infinity.
\end{abstract}
\maketitle
%
%
%
%%%%%%%%%%%%%%%%%%%%%%%%%%%%%%%%%%%%%%%%%%%%%%%%%%%%%%%%%%%%%%%%%%%%%%%%%%%%%%%
%%%%%%%%%%%%%%%%%%%%%%%%%%%%%%%%%%%%%%%%%%%%%%%%%%%%%%%%%%%%%%%%%%%%%%%%%%%%%%%
%%%%%%%%%%%%%%%%%%%%%%%%%%%%%%%%%%%%%%%%%%%%%%%%%%%%%%%%%%%%%%%%%%%%%%%%%%%%%%%

\section{Introduction}\label{sec-intro}

\noindent
A moving electron emits and absorbs electromagnetic radiation and is hence always
accompanied by a cloud of so-called soft photons.
Together with its photon cloud the electron behaves like
a particle whose mass is larger than the bare mass of the
electron. In an electrostatic potential it is thus easier
to bind an electron interacting with the quantized photon field
than the electron alone if the photon field were neglected.
It is well-known that these phenomena may be described mathematically
in the framework of non-relativistic (NR) quantum electrodynamics 
as follows:

First, we recall an effect called {\em enhanced binding}
due to the quantized radiation field.
We consider the non-relativistic electron Hamiltonian
\begin{equation}\label{carla0}
h_{\nr}(V):=-\tfrac{1}{2}\Delta_{\x}+V\,,
\end{equation}
acting in the Hilbert space $L^2(\RR^3)$. 
We suppose the potential $V=V_+-V_-$ to have
a short range negative part and recall
the following definition, which also applies to the
semi-relativistic operator $h_\sr(V)$ introduced in \eqref{def-hsr} below:

\begin{definition}\label{def-CCT}
Let $0\le V_+\in L_\loc^1(\R^3)$ and let $0\le V_-\not\equiv0$
satisfy $V_-\in L^{\nf{3}{2}}(\R^3)$, if $\sharp=\nr$, and
$V_-\in L^{\nf{3}{2}}\cap L^3(\R^3)$, if $\sharp=\sr$.
Define
$h_{\sharp}(V_\lambda)$ with $V_\lambda:=V_+-\lambda\,V_-$, $\lambda\ge0$, 
via a semi-bounded sum of quadratic forms.
We say $1$ is a {\em coupling constant threshold} for
$V_\lambda$, iff
\begin{enumerate}
\item[(1)] $\inf\spe(h_\sharp(V_+))=0$.
\item[(2)] $\inf\Spec(h_\sharp(V_\lambda))=0$, for $\lambda\in(0,1]$, and
$\inf\Spec(h_\sharp(V_\lambda))<0$, for $\lambda>1$.
\end{enumerate}
\end{definition}

\smallskip

\noindent
The existence of coupling constant thresholds
is a consequence of the variational principle and the
famous Cwikel-Lieb-Rosenbljum bound,
\begin{equation*}
\Tr\big[\id_{(-\infty,0)}(h_{\nr}(V_\lambda))\big]\le
\const\,\lambda^{\nf{3}{2}}\int_{\RR^3}V_-^{\nf{3}{2}}(\x)\,d^3\x\,.
\end{equation*} 
Here $\Tr$ denotes the trace and $\id_M(T)$ denotes
the spectral projection associated with some self-adjoint operator
$T$ and a Borel set $M\subset\RR$.
So the left hand side in the Cwikel-Lieb-Rosenbljum bound
counts all negative eigenvalues of $h_{\nr}(V_\lambda)$ including
multiplicities and, thus, has to be zero,
for sufficiently small $\lambda>0$ and $V_-\in L^{\nf{3}{2}}(\R^3)$.
Changing $V_-\not\equiv0$ 
by a multiplicative constant, if necessary, we may thus
achieve that $1$ is a coupling constant threshold.
Next, we take the interaction with the quantized radiation field
into account and consider the NR Pauli-Fierz operator
\begin{equation*}%\label{def-NRPF}
\Ham{\nr}(V):=
\tfrac{1}{2}(\vsigma\cdot(-i\nabla_{\x}+\ee\,\AAA))^2+V+\Hf\,,
\end{equation*}
where $\vsigma$ is a vector containing the Pauli matrices,
$\AAA\equiv\AAA_\Lambda$ is the quantized vector potential
in the Coulomb gauge with ultra-violet (UV) cutoff at $\Lambda>0$, 
$\Hf$ is the radiation field energy, and
$\ee\in\RR$ models the fine-structure constant.
We set
\begin{equation*}
\Egs{\nr}(V):=\inf\Spec(\Ham{\nr}(V))\,,
\end{equation*}
for any reasonable potential $V$. 
Back in our example $V_-\in L^{\nf{3}{2}}(\R^3)$ we say
that {\em enhanced binding occurs}, iff $1$ is a coupling constant
threshold for $V_\lambda$ and
$\Egs{\nr}(V_\lambda)$ is an eigenvalue of $\Ham{\nr}(V_\lambda)$, 
for some $\lambda<1$.
Notice that, if $1$ is a coupling constant threshold, then
$\inf\Spec(h_{\nr}(V_\lambda))$ cannot be an eigenvalue
of $h_{\nr}(V_\lambda)$, for any 
$\lambda<1$.
In order to observe this effect it suffices to
show that
\begin{equation}\label{michael1}
\Egs{\nr}(0)-\Egs{\nr}(V_\lambda)>0\,,\qquad \lambda\ge1-\delta\,,
\end{equation}
for some $\delta>0$.
In fact, according to \cite{GLL2001} the inequality in
\eqref{michael1} is a sufficient condition
for $\Egs{\nr}(V_\lambda)$ to be an eigenvalue.

In the past decade many mathematical articles dealt with
enhanced binding in NR quantum electrodynamics.
In the earliest one \cite{HiroshimaSpohn2001b} the dipole
approximation to the NR Pauli-Fierz model without spin
is considered and
enhanced binding is established, for all sufficiently large values of $\ee$.
The previous works on the full NR Pauli-Fierz model
\cite{BLV2005,BV2004,CEH,CH2004,ChenVV,HVV2003}
(with various additional conditions on $V_-$ and
sometimes without spin)
provide {\em complete} proofs of enhanced binding
under the assumptions that $|\ee|>0$ and/or $\Lambda>0$ be sufficiently 
small.
 Saying this we should, however, point out the
article \cite{ChenVV} where a general criterion 
for the occurrence of enhanced binding is established, 
namely existence of an eigenvalue at the bottom
of the spectrum of the fiber Hamiltonian corresponding
to total momentum $\V{0}$ of the translation invariant
electron-photon system as well as equality of this eigenvalue and
the self-energy of the electron.
While the conclusion proved in \cite{ChenVV} is always applicable,
no matter how big $|\ee|>0$ and $\Lambda>0$ are,
the existence of that eigenvalue has been
shown so far only for sufficiently small values of $|\ee|>0$ and/or $\Lambda>0$
\cite{Chen2008,CFP2009}.

The criterion established in \cite{ChenVV} can also be applied to
prove the {\em increase of binding energy} due to the quantized
radiation field. Here one assumes that the electronic Hamiltonian
with potential $V$
{\em does} have discrete eigenvalues below zero, which is 
henceforth again assumed to be the minimum of its essential spectrum.
Let $|e_V|$ be the absolute value of the lowest (strictly negative)
eigenvalue of $h_{\nr}(V)$. By definition the
{\em binding energy is increased}, iff
\begin{equation}\label{carla1}
\Egs{\nr}(0)-\Egs{\nr}(V)>|e_V|\,.
\end{equation}
In \cite{BCVV2010,Hainzl2003} this effect is observed on the 
basis of asymptotic expansions as $\ee$ goes to zero.
The estimates obtained in these articles
as well as the bounds on the shift of coupling constant
thresholds in \cite{BLV2005,CEH} come along with detailed
quantitative information on the coefficients in the expansions.
While these quantitative aspects
are interesting in their own right, it has always been expected
that there should exist entirely non-perturbative proofs
without any smallness assumptions on $\ee$ or $\Lambda$
covering, in particular, the physical value
$\ee^2\approx1/137$ no matter how big $\Lambda$ is chosen.
The first main achievement of the present article fills this
gap left open in the previous work. Namely, we prove 
enhanced binding and increase of binding energies 
in the NR Pauli-Fierz model, 
{\em for all values of $|\ee|>0$ and $\Lambda>0$}.
To this end we employ variational arguments similar to
those in \cite{ChenVV} with the crucial
difference, however, that our trial functions are constructed by means of
{\em minimizing sequences} of fiber Hamiltonians (instead of eigenvectors)
whose detailed properties are unknown a priori.
The main new technical difficulty is to provide estimates 
holding uniformly along such minimizing sequences. 
These estimates show that actually no non-trivial a priori
knowledge on the mass shell and no deep results on the existence
of ground states of fiber Hamiltonians are required to give a
qualitative discussion of enhanced binding and increased binding energies.
In the discussion of enhanced binding
we are also able to relax earlier assumptions on the short range potential 
whose (non-vanishing) negative part need not satisfy any other
condition than being in $L^{\nf{3}{2}}(\R^3)$.
The mildest condition on the local singularities stated in the quoted
literature is $V_-\in L_\loc^4(\R^3)$ \cite{BLV2005}.
As a technical prerequisite some information
on the convergence of electronic
eigenfunctions to threshold energy states is needed here (see also \cite{ChenVV}).
Corresponding results are supplied by \cite{SorStock}
(also in the semi-relativistic case discussed below)
for bounded and integrable potentials.
We shall push these results 
a little forward to the broader class of potentials considered here.

The second purpose of our paper is to study the enhancement of binding
and the increase of binding energies in the {\em semi}-relativistic (SR)
Pauli-Fierz model, whose mathematical analysis has been 
initiated in \cite{FGS2001,MiyaoSpohn2009}.
This model is obtained by replacing the symbol
$\tfrac{1}{2}|\xi|^2$ of the kinetic energy in the NR model
by its relativistic analog, $\sqrt{|\xi|^2+1}-1$.
Thus, the electron Hamiltonian reads
\begin{equation}\label{def-hsr}
h_{\sr}(V):=\sqrt{1-\Delta_{\x}}-1+V\,,
\end{equation}
and the full SR Pauli-Fierz Hamiltonian for the interacting system is still
obtained via minimal coupling to the quantized radiation field,
\begin{equation}
\Ham{\sr}(V):=\sqrt{(\vsigma\cdot(-i\nabla_{\x}+\ee\,\AAA))^2+1}-1
+V+\Hf\,.
\end{equation}
In the SR case it follows from
\cite{Cwikel1977,Daubechies1983} that
\begin{equation*}
\Tr\big[\id_{(-\infty,0)}(h_{\sr}(V_\lambda))\big]\le
\const\,\int_{\RR^3}\big((\lambda \,V_-(\x))^{\nf{3}{2}}+
(\lambda \,V_-(\x))^3\big)\,d^3\x\,.
\end{equation*}
So, we again expect to observe an enhanced binding,
for non-zero $V_-$ belonging to
$L^{\nf{3}{2}}(\RR^3)$ (this condition is due to
$\sqrt{|\xi|^2+1}-1\sim\tfrac{1}{2}|\xi|^2$ for small $|\xi|$)
as well as to $L^3(\RR^3)$ (due to 
$\sqrt{|\xi|^2+1}-1\sim|\xi|$ for large $|\xi|$).
In fact, criteria for the existence of ground states for the
relevant class of short range potentials are given in \cite{KMS2012}.
A related problem is treated
in \cite{HiroshimaSasaki2012}
where $N$ relativistic spin-less particles in a short range potential
interacting via a linearly coupled bosonic field are considered.
If a scaling parameter in front of the creation and annihilation operators
is sufficiently large
(weak coupling limit)
and if the coupling constant in front of the interaction lies
in a certain bounded interval,
then the authors are able to show existence of
a unique ground state of the total Hamiltonian.
A non-strict inequality analogous to \eqref{carla1}
has been obtained in the SR case first in \cite{HiroshimaSasaki2010}.

As a byproduct of our analysis we verify that the renormalized
electron mass in the SR Pauli-Fierz model is always
strictly larger than the bare mass of the electron, as soon
as it may be defined (as the inverse second radial derivative
of the mass shell at zero).
For small $|\ee|>0$, depending on $\Lambda$, 
the {\em existence} of the renormalized
electron mass in the SR Pauli-Fierz model has been proved recently
by the present authors in \cite{KM2012a}.
In another application of our ideas 
we determine the linear leading order term in the
asymptotics of the ground state energy of the free
SR Pauli-Fierz operator, as $\Lambda$ goes to infinity.
Asymptotically linear upper and lower bounds
on the self-energy have been obtained
earlier in \cite{LiebLoss2000}.

The organization of this article is given as follows.
In Subsection~\ref{ssec-models} we introduce all operators
studied here more precisely.
All of our main results are stated precisely in Subsection~\ref{ssec-mainres}.
In Section~\ref{Abb4} we develop the crucial technical estimates
used to derive our main theorems. In
Sections~\ref{SecNR} and~\ref{SecSR} we apply them to the
NR and SR models, respectively.
In the appendix
we recall some Birman-Schwinger principles
and extend some results from \cite{SorStock} on the convergence of
eigenfunctions to threshold energy states.

%%%%%%%%%%%%%%%%%%%%%%%%%%%%%%%%%%%%%%%%%%%%%%%%%%%%%%%%%%%%%%%%%%%%%%%%%%%%%%
%%%%%%%%%%%%%%%%%%%%%%%%%%%%%%%%%%%%%%%%%%%%%%%%%%%%%%%%%%%%%%%%%%%%%%%%%%%%%%
%%%%%%%%%%%%%%%%%%%%%%%%%%%%%%%%%%%%%%%%%%%%%%%%%%%%%%%%%%%%%%%%%%%%%%%%%%%%%%

\section{Models and main results}
\label{sec-results}

%%%%%%%%%%%%%%%%%%%%%%%%%%%%%%%%%%%%%%%%%%%%%%%%%%%%%%%%%%%%%%%%%%%%%%%%%%%%

\subsection{Definition of the models}
\label{ssec-models}

\subsubsection{Bosonic Fock space}

\noindent
First, we fix some notation for operators acting in 
the state space of the photon
field, the bosonic Fock space, $\Fock_b$.
In what follows an italic $k$ always denotes a tuple
$k=(\V{k},\lambda)\in\RR^3\times\ZZ_2$ and $\cA$ is a
non-void open subset of $\RR^3$.
(In applications we encounter the examples $\cA=\RR^3$ or
$\cA=\{|\V{k}|>m\}$ with $m>0$.)
For every $n\in\NN$, let $\cS_n$ denote the orthogonal
projection in $L^2 ((\cA\times \Z_2)^n)$ onto the space of 
permutation symmetric functions. That is,
\begin{equation*}
(\mathcal{S}_n\,\psi^{(n)})(k_1,\ldots,k_n):= 
\frac{1}{n!} \sum_\pi \psi^{(n)}(k_{\pi(1)},\ldots,k_{\pi(n)})\,,
\end{equation*}
almost everywhere, for $\psi^{(n)}\in L^2 ((\cA\times \Z_2)^n)$,
the sum running over all permutations of $\{1,\ldots,n\}$.
Then the bosonic Fock space modeled over the one photon
Hilbert space $\mathfrak{h}:=L^2 (\cA\times \Z_2,dk)$,
$\int dk:= \sum_{\lambda\in \Z_2}\,\int_{\cA} d^3\V{k}$,
is the direct sum
$$
\Fock_b:= \bigoplus_{n=0}^\infty \Fock_b^{(n)},\quad\textrm{with}\;\;
\Fock_b^{(0)}:=\C\,,\;\; \Fock_b^{(n)}:=
\mathcal{S}_n L^2 ((\cA\times \Z_2)^n)\,,\;\;n\in\NN\,.
$$
The vector $\Omega:=\{1,0,0,\ldots\,\}\in\Fock_b$
is called the vacuum.
We denote by $\sC$ the dense subspace of all
$\{\psi^{(n)}\}_{n=0}^\infty\in\Fock_b$ such that only
finitely many $\psi^{(n)}$ are non-zero and each
$\psi^{(n)}$, $n\in\NN$, has a compact support.

For $f\in\mathfrak{h}$, let $a^\dagger(f)$ and $a(f)$
denote the standard bosonic creation and annihilation operators,
respectively. 
Setting 
$a^\dagger(f)^{(n)}\,\psi^{(n)}:= 
(n+1)^{\nf{1}{2}}\mathcal{S}_{n+1}\,(f\otimes \psi^{(n)})$,
for $\psi^{(n)}\in\Fock_b^{(n)}$, $n\in\NN_0$,
the creation operator is the
closed operator given by
$a^\dagger(f)\,\psi:= \{0,a^\dagger(f)^{(0)}\,\psi^{(0)}, 
a^\dagger(f)^{(1)}\,\psi^{(1)},\ldots\;\}$,
for all $\psi=\{ \psi^{(n)}\}_{n=0}^\infty\in \Fock_b$
in its maximal domain, and $a(f):= a^\dagger(f)^*$.
The following {canonical commutation relations (CCR)} are satisfied
on a suitable dense domain (e.g., on $\sC$),
\begin{equation*}
[a(f),\,a(g)]=0\,,\quad [a^\dagger(f),\,a^\dagger(g)]=0\,,
\quad [a(f),\,a^\dagger(g)]= \sps{f}{g}\,\id\,,
\end{equation*}
for $ f,g \in \mathfrak{h}$.
The second quantization of
a real-valued Borel function, $\vk$, on $\cA$,
is the self-adjoint operator in $\Fock_b$ given by
$d\Gamma(\vk)\!\!\upharpoonright_{\Fock^{(0)}}:=0$ and
\begin{equation*}
d\Gamma(\vk)\!\!\upharpoonright_{\Fock^{(n)}}\psi^{(n)}(k_1,\ldots,k_n)
:=(\vk(\V{k}_1)+\cdots+\vk(\V{k}_n))\,\psi^{(n)}(k_1,\ldots,k_n)\,,
\end{equation*}
for $n\in\NN$, $k_j=(\V{k}_j,\lambda_j)$. 
The multiplication operator
$d\Gamma(\vk)$ is defined on its maximal domain.
For instance, the {field energy operator}, $\Hf$, and the
{field momentum operator}, $\pf$, are defined by
\begin{equation*}
\Hf:= d\Gamma(\omega),\qquad 
\pf:=d\Gamma(\vmu):=\big(d\Gamma(\mu_1),d\Gamma(\mu_2),d\Gamma(\mu_3)\big)\,.
\end{equation*}
The physically relevant choices of $\omega$ and
$\vmu=(\mu_1,\mu_2,\mu_3)$
are given in Example~\ref{ex-Gphys} below.
Henceforth, we shall, however, only assume that
$\omega,\mu_1,\mu_2,\mu_3:\cA\to\RR$ are measurable
such that
\begin{equation}\label{hyp-omega}
0<\omega(\V{k})\le d\,(|\V{k}|+1)\,,\qquad
|\vmu(\V{k})|\le d\,\omega(\V{k})\,,\qquad\textrm{a.e.}\;\V{k}\in\cA\,,
\end{equation}
for some $d>0$.
The following standard estimates 
shall be useful later on,
\begin{align} \nonumber
\|a(f_1)\ldots a(f_n)\,\psi\|
&\le 
\|f_1\|_{\nf{1}{2}}\dots\|f_n\|_{\nf{1}{2}}\,\|\Hf^{\nf{n}{2}}\,\psi\|\,,
\\\label{StandardEstimate}
\|a^\dagger(f)\,\psi\|^2
&\le\|f\|^2_{\nf{1}{2}}\,\|\Hf^{\nf{1}{2}}\,\psi\|^2+\|f\|^2\,\|\psi\|^2,
\end{align}
with $\|f\|_{\nf{1}{2}}:=\| \omega^{-\nf{1}{2}}\,f\|$, for all 
$f,f_j\in\mathfrak{h}$ and $\psi\in\Fock_b$
such that the right hand sides are finite.
For every $f\in\mathfrak{h}$, the operator
$2^{-\nf{1}{2}}(a^\dagger(f)+a(f))$ is essentially self-adjoint on $\sC$.
We denote its self-adjoint extension by $\vp(f)$ and
write $\vp(\V{f}):=(\vp(f_1),\vp(f_1),\vp(f_3))$, for a triple of photon 
wave functions
$\V{f}=(f_1,\,f_2,\,f_3)\in \mathfrak{h}^3$.

\subsubsection{Fiber Hamiltonians}

\noindent
We next define Hamiltonians acting in
$\CC^2\otimes\Fock_b$
which is henceforth referred to as the fiber Hilbert space.
For reasons illustrated by Example~\ref{ex-Gmeps} below
we work with general conditions on the dispersion relation
$\omega$, the vector field $\vmu$, and the coupling function $\V{G}$: 

\begin{hypothesis}\label{hyp-G}
$\omega:\cA\to\RR$, $\vmu:\cA\to\RR^3$, and
$\V{G}:\cA\times\ZZ_2\to\RR^3$ are measurable and
satisfy \eqref{hyp-omega} and
\begin{equation}\label{eq-hyp-G}
\|\V{G}\|\ge g\,,\qquad
\int\omega^\ell\,|\V{G}|^2\le d^2,\quad
\ell\in\{-1,0,\ldots,13\}\,,
\end{equation}
respectively, for some $d\ge1$, $g>0$.
Moreover, 
\begin{equation}\label{hyp-r}
\|\id_{\{\omega\le\delta\}}\,\V{G}\|
+\|\id_{\{\omega\le\delta\}}\,\omega^{-\nf{1}{2}}\,\V{G}\|
\le r(\delta)\,,\quad\delta>0\,,
\end{equation}
for some non-negative function $r:(0,\infty)\to\RR$
with $r(\delta)\to0$, $\delta\downarrow0$.
\hfill$\Diamond$
\end{hypothesis}

\smallskip

\noindent
The somewhat mysterious bound $\ell\le13$ in \eqref{eq-hyp-G}
is due to the application of a certain higher order estimate
in Lemma~\ref{le-hoe} below. The function $r$ is introduced
in order to treat several choices of $\V{G}$ at the same time
and to quantify their infra-red behavior in a uniform fashion.

\begin{example}\label{ex-Gphys}
(i) Physically relevant choices of $(\omega,\vmu,\V{G})$
fulfilling Hypothesis~\ref{hyp-G}
are given by $\omega(\V{k}):=|\V{k}|$, $\vmu(\V{k}):=\V{k}$,
for $\V{k}\in\cA:=\RR^3$,
and
\begin{equation}\label{G-nr}
\V{G}=\vr(|\V{k}|)\,\veps(\mr{\V{k}},\lambda)\,,
\qquad\mr{\V{k}}:=\V{k}/|\V{k}|\,,
\end{equation}
for almost every $\V{k}$ and $\lambda\in\ZZ_2$,
where $\vr$ is some measurable real function with
$$
0<\int_0^\infty(t+t^{15})\,\vr^2(t)\,dt<\infty\,,
$$
and $\{\mr{\V{k}},\veps(\mr{\V{k}},0),\veps(\mr{\V{k}},1)\}$
is an oriented orthonormal basis of $\RR^3$, for a.e. $\mr{\V{k}}\in
S^2$.

\smallskip

\noindent(ii)
A common special case of (i) is given by
$\V{G}:=\V{G}^{\ee}_\Lambda$ with 
\begin{align}\label{def-Gphys}
\V{G}^{\ee}_\Lambda(\V{k},\lambda)
:=(2\pi)^{-\nf{3}{2}}\,\ee\,|\V{k}|^{-\nf{1}{2}}\,
\id_{|\V{k}|<\Lambda}\,\veps(\mr{\V{k}},\lambda)\,,
\end{align}
where $\ee\in\RR\setminus\{0\}$ and $\Lambda>0$ is an UV cutoff parameter.
\hfill$\Diamond$
\end{example}

\smallskip

\noindent
The main reason why we introduce the quantities $d$, $g$, and $r$
in the above hypothesis is the following example.
The modified versions of the physical choices of $\omega$, $\vmu$,
and $\V{G}$ defined there appear in proofs of the
existence of ground states; see, for instance, \cite{KMS2009a} and
the proof of Corollary~\ref{GroundStates} below.

\begin{example}\label{ex-Gmeps}
Let $\omega$, $\vmu$, and $\V{G}$ be as in Example~\ref{ex-Gphys}(ii).

\smallskip

\noindent
(i) Let $\V{G}_m:=\id_{\{\omega\ge m\}}\,\V{G}^{\ee}_\Lambda$ and
$m_0\in(0,\Lambda)$. Trivially, all
$(\omega,\vmu,\V{G}_m)$ with $0<m\le m_0$ fulfill
Hypothesis~\ref{hyp-G} with the same suitable choices of $d$, $g$, $r$.

\smallskip

\noindent(ii) Pick some $m>0$ and replace $\RR^3$ by
$\cA_m:=\{|\V{k}|> m\}$
in Example~\ref{ex-Gphys}.
Set $Q(\vnu):=\vnu+(-1/2,1/2]^3$ and
$Q_\ve(\vnu):=(\ve\,Q(\vnu))\cap\cA_m$,
for all $\ve>0$ and $\vnu\in\ZZ^3$.
Set
$\omega_\ve\!\!\upharpoonright_{Q_\ve(\vnu)}:=\inf_{Q_\ve(\vnu)}\omega$,
let
$\vmu_\ve\!\!\upharpoonright_{Q_\ve(\vnu)}$ be constantly equal to
some
arbitrary vector in $\ol{Q}_\ve(\vnu)$, and let
$\V{G}_\ve\!\!\upharpoonright_{Q_\ve(\vnu)}$
be constantly equal to the average of $\V{G}^{\ee}_\Lambda$ over $Q_\ve(\vnu)$.
Then we find ($m$-dependent) $\ve_0>0$,
$d$, $g$, and $r$ such that all $(\omega_\ve,\vmu_\ve,\V{G}_\ve)$,
$0<\ve<\ve_0$, fulfill Hypothesis~\ref{hyp-G} with these 
fixed choices of $d$, $g$, and $r$.
\hfill$\Diamond$
\end{example}

\smallskip

\noindent
Let $\vsigma := (\sigma_1,\,\sigma_2,\,\sigma_3)$ 
be the triple of Pauli spin matrices
%
%\begin{equation*}
%\sigma_1:= 
%\begin{pmatrix}
%0 & 1\\
%1& 0
%\end{pmatrix},
%
%\qquad
%\sigma_2:= 
%\begin{pmatrix}
%0 & -\imath \\
%\imath & 0
%\end{pmatrix},
%
%\qquad 
%\sigma_3:=
%\begin{pmatrix}
%1 & 0 \\
%0 & -1
%\end{pmatrix}.
%\end{equation*}
%
and write $\vsigma\cdot\V{v}:=\sigma_1\,v_1+\sigma_2\,v_2+\sigma_3\,v_3$,
for a vector $\V{v}=(v_1,v_2,v_3)$ whose entries are complex numbers
or suitable operators.
For every $\V{p}\in\RR^3$, we then define
\begin{align*}
\V{v}(\V{p}):=\p-\pf+\vp(\V{G})\,,\qquad
w(\p):=\vsigma\cdot\V{v}(\V{p})\,.
\end{align*}
Applying Nelson's commutator theorem with test operator $\Hf+1$ we
verify that
$w(\p)$ is essentially self-adjoint on any core of $\Hf$.
We denote its self-adjoint closure starting from $\sC$
again by the same symbol and
define
\begin{equation*}
\hat{\tau}_\nr(\p):=\tfrac{1}{2}\,w(\p)^2\,,\qquad
\hat{\tau}_\sr(\p):=\sqrt{w(\p)^2+1}-1\,,
\end{equation*}
by means of the spectral calculus.
Next, we define fiber Hamiltonians
\begin{align*}
\HamF{\sharp}(\p)&:=\hat{\tau}_\sharp(\p)+\Hf\,,
\qquad\sharp\in\{\nr,\sr\}\,,
\end{align*}
as Friedrichs extensions starting from $\sC$.
For $\V{G}=\V{0}$, we denote them by
\begin{align*}
\HamF{\nr}^0(\p):=\tfrac{1}{2}\,(\p-\pf)^2+\Hf\,,\qquad
\HamF{\sr}^0(\p):=\sqrt{(\p-\pf)^2+1}-1+\Hf\,.
\end{align*}
It is known \cite[Lemma~2.2(ii)]{KM2012a} that
$\dom(\HamF{\sr}(\p))=\dom(\Hf)$ and, for all $\ve>0$,
\begin{align}\label{pia1}
\big\|(\HamF{\sr}(\p)-\HamF{\sr}^0(\p))\,\vp\big\|
&\le\ve\,\|\HamF{\sr}^0(\p)\,\vp\|+\const(\ve,d)\,\|\vp\|\,,
\quad\vp\in\dom(\Hf)\,.
\end{align}
In particular, $\sC$ is a core for $\HamF{\sr}(\p)$.
The mass shells are defined by
\begin{align*}
E_\sharp(\p):=\inf\Spec(\HamF{\sharp}(\p))\,,\qquad\p\in\RR^3.
\end{align*}

\subsubsection{Total Hamiltonians}\label{ssec-tot}

\noindent
Finally, we introduce the Hamiltonians generating the time evolution 
of the combined electron-photon system.
The total Hilbert space is
$$
\HR:=L^2(\RR^3,\CC^2)\otimes\Fock_b
=\int_{\RR^3}^\oplus\CC^2\otimes\Fock_b\,d^3\x\,.
$$
The quantized vector potential, $\AAA$, is the triple of operators
given by
$$
\AAA:=\int_{\RR^3}^\oplus\id_{\CC^2}\otimes\vp(e^{-i\vmu\cdot\V{x}}\,\V{G})\,
d^3\V{x}\,.
$$
We drop all trivial tensor factors in the notation in what follows
($-i\nabla_\x\equiv-i\nabla_\x\otimes\id$, $\Hf\equiv\id\otimes\Hf$, etc.) and 
define
$$
\ww:=\vsigma\cdot(-i\nabla_{\V{x}}+\ee\,\AAA)
$$
on the domain 
$$
\sD_1:=\dom(-\Delta_{\V{x}})\cap\dom(\Hf)
$$ 
to begin with. 
An application of Nelson's commutator theorem shows that
$\ww$ is essentially self-adjoint on any core of 
$-\Delta_{\V{x}}+\Hf$ and in particular on $\sD_1$ and 
on the algebraic tensor product
$$
\core:=C_0^\infty(\RR^3,\CC^2)\otimes\sC\,.
$$
Denoting the closure of $\ww$ again by the same symbol
we define
\begin{equation}\label{def-tau}
\tau_\nr:=\tfrac{1}{2}\,\ww^2,\qquad
\tau_\sr:=\sqrt{\ww^2+1}-1\,,
\end{equation}
by the spectral calculus.

Next, let $\sharp\in\{\nr,\sr\}$, recall the notation  
\eqref{carla0} and \eqref{def-hsr},
and assume that $V\in L^1_\loc(\RR^3,\RR)$ satisfies
\begin{equation}\label{def-cV}
c_V:=-\inf\big\{\sps{\psi}{h_\sharp(V)\,\psi}\,:\:\psi\in
C_0^\infty(\R^3)\,,\:
\|\psi\|=1\big\}<\infty\,.
\end{equation}
Then it is known that 
\begin{equation}\label{lb-Ham}
\Ham{\sharp}(V):=\tau_\sharp+V+\Hf\ge-c_V-\const\,d^2
\end{equation}
in the sense of quadratic forms on $\sD$, for some
universal constant $\const>0$.
In the NR case this is a well-known consequence of
diamagnetic inequalities (see, e.g., the review in \cite{KMS2012})
and relative bounds on the magnetic field
with respect to $\Hf$. In the SR case \eqref{lb-Ham} 
follows from \cite[Theorem~3.4]{KMS2012}.
Therefore, the quadratic forms of
$h_\sharp(V)$ and
$\Ham{\sharp}(V)$ equipped with the domains 
$C_0^\infty(\RR^3)$ and $\sD$, 
respectively, are closable. We denote the self-adjoint operators
representing the closures of these forms
again by the same symbols.
According to \cite{Hiroshima2000} (in the NR case)
and \cite{KMS2010} (in the SR case) the 
operators $\Ham{\sharp}(0)$ are essentially
self-adjoint on $\sD$.
The domain of $\Ham{\sr}(0)$ is
$\dom((-\Delta_\x)^{\nf{1}{2}}+\Hf)$ \cite{KMS2010}.
We set
$$
\Egs{\sharp}(V):=\inf\Spec(\Ham{\sharp}(V))\,.
$$

\subsubsection{Fiber decompositions}\label{sssec-fiber}

\noindent
Setting
\begin{equation}\label{def-Uq}
U_\q:=e^{-i(\q-\pf)\cdot\x}
\end{equation}
and denoting the (partial) Fourier transform with respect to
$\x$ by $\cF$,
we observe that
$U_\q^*\cF^*\int_{\RR^3}^\oplus w(\q+\p)\,d^3\p\,\cF\,U_\q$
is a self-adjoint extension of ${\ww}\!\!\upharpoonright_{\sD}$
and, hence, equal to ${\ww}$.
Using \cite[Theorem~XIII.85]{ReedSimonIV} in the second step
we conclude that
$$
f({\ww})=U_\q^*\cF^*f\Big(\int_{\RR^3}^\oplus w(\q+\p)\,d^3\p\Big)\,\cF\,U_\q
=U_\q^*\cF^*\int_{\RR^3}^\oplus f(w(\q+\p))\,d^3\p\;\cF\,U_\q\,,
$$
where $f$ is $x^2$ or a bounded Borel function.
Writing $f(x)=(x^2+1)^{\nf{1}{2}}$ as $(x^2+1)\,(x^2+1)^{-\nf{1}{2}}$
we see that this choice of $f$ is allowed, too.
But then it follows that
$U_\q^*\cF^*\int_{\RR^3}^\oplus \HamF{\sharp}(\q+\p)\,d^3\p\,\cF\,U_\q$ is a
self-adjoint
extension of the essentially self-adjoint operator 
$\Ham{\sharp}(0)\!\!\upharpoonright_{\sD}$, whence
we have the fiber decomposition
\begin{equation}\label{fib-dec}
\cF\,U_\q\,\Ham{\sharp}(0)\,U_\q^*\cF^*
=\int_{\RR^3}^\oplus \HamF{\sharp}(\q+\p)\,d^3\p\,.
\end{equation}

\begin{lemma}\label{le-mass-shell}
If $(\omega,\vmu,\V{G})$ fulfill Hypothesis~\ref{hyp-G},
$\sharp\in\{\nr,\sr\}$, and $\pmax>0$,
then there exist $\emax\equiv\emax(\pmax,d)$ and
$\pmax_*\equiv\pmax_*(d)>0$,
such that
\begin{equation}\label{joice}
\sup_{|\p|\le\pmax}\E{\sharp}\le\emax\,,\qquad
\Egs{\sharp}(0)=\essinf{|\p|\le\pmax_*}\,\E{\sharp}(\p)\,.
\end{equation}
\end{lemma}

\begin{proof}
On account of \eqref{fib-dec} and \cite[Theorem~XIII.85]{ReedSimonIV}
we have $\Egs{\sharp}(0)=\tessinf{\RR^3}\,\E{\sharp}$.
In the SR case the bounds
$\tfrac{1}{2}|\p|-\const(d)\le
E_\sr(\p)\le\tfrac{3}{2}|\p|+\const(d)$,
follow in a straightforward fashion
from \eqref{pia1} and imply \eqref{joice}.
Concerning the NR case,
the upper bound $E_\nr(\p)\le\p^2/2+ d^2$
follows immediately by testing with vectors in the 
vacuum sector. Finally,
for $\delta\in(0,1)$ sufficiently
close to $1$, we have, as quadratic forms on $\sC$,
\begin{align*}
\HamF{\nr}(\p)&\ge(1-\delta)\,\tfrac{1}{2}(\p-\pf)^2+(1-\delta^{-1})\,\tfrac{1}{2}\,\vp(\V{G})^2+\Hf
\\
&\ge(1-\delta)\,\tfrac{1}{2}(\p-\pf)^2+\const\,(1-\delta^{-1})\,d^2\,(\Hf+1)+\Hf
\\
&\ge\const(d)^{-1}\,\HamF{\nr}^0(\p)-\const'(d)\,.
\end{align*}
But \eqref{hyp-omega} implies
$d\,\HamF{\nr}^0(\p)\ge\id_{|\p|<1}\p^2/2+\id_{|\p|\ge1}(|\p|-1/2)$
and we again obtain \eqref{joice}.
\end{proof}

\smallskip

\noindent
The bounds in the previous proof are by no means optimal. Moreover, one
can always show continuity of the mass shells and under
physically reasonable assumptions they are supposed to attain their
unique minimum at $\V{0}$. We gave a very simple, self-contained 
proof since more detailed information than in Lemma~\ref{le-mass-shell}
would not lead to any relevant simplifications in our proofs.

%%%%%%%%%%%%%%%%%%%%%%%%%%%%%%%%%%%%%%%%%%%%%%%%%%%%%%%%%%%%%%%%%%%%%%%%%%%%

\subsection{Main results}
\label{ssec-mainres}

\noindent
The first main results of the present paper are summarized
in the following theorem.
As already stressed above, what is crucial here is that
Theorem~\ref{thm-binding} applies to the physical
Example~\ref{ex-Gphys}(ii) without any restrictions on
$|\ee|,\Lambda>0$.
In the SR case the implications of Theorem~\ref{thm-binding}
are new also when $|\ee|,\Lambda>0$ are small
in that example. Recall the definitions of
$h_\sharp(V)$, $\Ham{\sharp}(V)$, and $\Egs{\sharp}(V)$ in
Sub-subsection~\ref{ssec-tot}.

\begin{theorem}[{\bf Increased and enhanced binding}]
\label{thm-binding}
Assume that $(\omega,\vmu,\V{G})$ fulfill Hypothesis~\ref{hyp-G}
with parameters
$d$, $g$, and $r$.
In the case $\sharp=\nr$ assume in addition that
$\omega$, $\vmu$, and $\V{G}$ are as in Example~\ref{ex-Gphys}(i). 

\smallskip

\noindent
(a)
Let $V\in L^1_\loc(\RR^3,\RR)$ satisfy \eqref{def-cV}. If
\begin{equation}\label{knut}
\inf\spe(h_{\sharp}(V))=0\,,\qquad
e_{\sharp}(V):=\inf\Spec(h_{\sharp}(V))<0\,,
\end{equation}
then we find some $c\equiv c(d,g,r,V)>0$ such that
\begin{equation*}
\Egs{\sharp}(0)-\Egs{\sharp}(V)-e_{\sharp}(V)\ge c\,.
\end{equation*}
(b) 
If $\sharp=\nr$, let 
$0\le V_-\in L^{\nf{3}{2}}(\R^3)$, $V_-\not\equiv0$, 
and $0\le V_+=V_{+,1}+V_{+,2}$ with
$V_{+,1}\in L^{\nf{3}{2}}(\R^3)$ and $V_{+,2}\in L^\infty(\R^3)$.
If $\sharp=\sr$, let
$0\le V_-\in L^{\nf{3}{2}}\cap L^3(\R^3)$, $V_-\not\equiv0$, 
and $0\le V_+\in L^1_\loc(\R^3)$.
Set
$V_\mu:= V_+-\mu V_- $, $\mu>0$, and
assume that $1$ is a coupling constant threshold for $V_\mu$.
Then there exist $c,\delta>0$, both depending only
on $d$, $g$, $r$, and $V_\pm$, such that
\begin{equation*}
\Egs{\sharp}(0)-\Egs{\sharp}(V_\mu)\ge c\,,
\qquad\mu\ge1-\delta\,.
\end{equation*}
\end{theorem}

\begin{proof}
The proofs of this theorem in the cases $\sharp=\nr$ and
$\sharp=\sr$ are given in Sections~\ref{SecNR} and~\ref{SecSR},
respectively.
\end{proof}

\smallskip

\noindent
In the NR case we restrict ourselves to the situation of
Example~\ref{ex-Gphys}(i)
since we exploit the rotation invariance of the free Hamiltonian
in that case. 

\begin{example}\label{ex-KMS3}
In the situation of Example~\ref{ex-Gmeps}(i) (resp. (ii) with fixed $m>0$)
label all quantities defined by means of $(\omega,\vmu,\V{G}_m)$ 
(resp. $(\omega_\ve,\vmu_\ve,\V{G}_{\ve})$) by a superscript
$m$ (resp. $\ve$). 
If $V$ is as in Theorem~\ref{thm-binding}(a) or
$V=V_+-\mu\,V_-$, $\mu\ge1-\delta$, with $V_\pm$ and $\delta$
as in Theorem~\ref{thm-binding}(b), then we find
$m_0,\ve_0>0$ with
\begin{equation}\label{hyp-KMS3}\;
\inf_{0<m\le m_0}(\Egs{\sr}^{m}(0)-\Egs{\sr}^m(V))>0\,,\qquad
\inf_{0<\ve\le \ve_0}(\Egs{\sr}^{\ve}(0)-\Egs{\sr}^\ve(V))>0\,.
\qquad\Diamond
\end{equation}
\end{example}

\begin{corollary}[{\bf Existence of ground states}]
  \label{GroundStates}
Assume we are in the situation of Example~\ref{ex-Gphys}(ii)
with arbitrary values of $|\ee|,\Lambda>0$.

\smallskip

\noindent
(a) 
In the case $\sharp=\nr$, let 
$V\in L^1_\loc(\R^3,\R)$
be infinitesimally form-bounded with respect to $-\Delta_\x$. In the case
$\sharp=\sr$, let $V\in L^2_\loc(\R^3,\R)$ be relatively form-bounded
with respect to $(-\Delta_\x)^{\nf{1}{2}}$ with relative form
bound $<1$. If \eqref{knut} holds 
and $V(\V{x})\to0$, as $|\V{x}|\to\infty$, then
$\Ham{\nr}(V)$ (resp. $\Ham{\sr}(V)$) 
has normalizable ground state
eigenvectors. 

\smallskip

\noindent
(b)
If the potentials $V_\pm$ are as in
Theorem~\ref{thm-binding}(b)
with $V_\pm(\V{x})\to0$, $|\V{x}|\to\infty$, 
then we find some $\delta>0$ such that 
$\Ham{\nr}(V_\mu)$ (resp. $\Ham{\sr}(V_\mu)$) 
has normalizable ground state
eigenvectors,
for every $\mu\ge 1-\delta$.
\end{corollary}

\begin{proof}
In the NR the case both (a) and (b) follow from \cite[Theorem~2.1]{GLL2001}
according to which
binding, i.e. the inequality $\Egs{\nr}(V)<\Egs{\nr}(0)$, 
implies the existence of ground states.
In the SR case (a) and (b) can be proved by straightforward
modifications of the arguments in \cite{KMS2009a} where
the Coulomb potential is treated.
The details are worked out in \cite{KMS2012}.
In fact, according to \eqref{hyp-KMS3} 
the uniform binding conditions
postulated in Hypothesis~6.6 of \cite{KMS2012}
are fulfilled and, hence, (a) and (b)
are special cases of \cite[Theorem~8.1]{KMS2012}.
\end{proof}

\smallskip

\noindent
It is possible to prove
the existence of ground states of $\Ham{\sr}(V_{\gamma})$
with $V_{\gamma}(\x):=-\gamma/|\x|$ also 
in the critical case $\gamma=\nf{2}{\pi}$ 
where the relative form bound of $V_{\nf{2}{\pi}}$
with respect to $(-\Delta_\x)^{\nf{1}{2}}$ is equal to one 
\cite{KM2011}.
(For $\gamma>\nf{2}{\pi}$, the quadratic form of $\Ham{\sr}(V_{\gamma})$
is unbounded below \cite{KMS2009a}.)
Note that Theorem~\ref{thm-binding}(a) applies to $\Ham{\sr}(V_{\nf{2}{\pi}})$.
It turns out that the decay rate of the spatial 
$L^2$-exponential localization of 
ground state eigenvectors of $\Ham{\sharp}(V)$ is strictly
bigger
than the decay rate of the electronic eigenfunctions (if any).

\begin{corollary}[{\bf Increase of
    localization}]\label{IncreaseLocalization}
Assume we are in the situation of Example~\ref{ex-Gphys}(ii)
with arbitrary $|\ee|,\Lambda>0$.
Let $V$ satisfy the conditions of
Corollary~\ref{GroundStates}(a) or (b)
or suppose $V(\x)=-(\nf{2}{\pi})/|\x|$.
Let $\Phi_\sharp$ be a ground state eigenvector of $\Ham{\sharp}(V)$.
Then, in the NR case, 
\begin{equation}\label{exp-NR}
\forall\,\beta>0\::\quad\beta^2/2< \Egs{\nr}(0)-\Egs{\nr}(V)\;\;
\Rightarrow\;\;
e^{\beta\, |\x|}\Phi_{\nr}\in\HR\,,
\end{equation}
and in the SR case
\begin{equation}\label{exp-SR}
\forall\,\beta\in(0,1)\::\quad
1-(1-{\beta}^2)^{\nf{1}{2}}< \Egs{\sr}(0)-\Egs{\sr}(V)\;\;
\Rightarrow\;\;
e^{{\beta}\, |\V{x}|}\,\Phi_{\sr}\in\HR\,.
\end{equation}
\end{corollary}

\begin{proof}
The bound \eqref{exp-NR} follows from \cite{GLL2001}; see also \cite{Griesemer2004}.
The bound \eqref{exp-SR} is a special case of
\cite[Theorem~5.1]{KMS2012}
where the decay rates found in an earlier 
localization estimate \cite{MatteStockmeyer2009a} are improved.
The crucial observation that led to the decay rates in \eqref{exp-SR} has been made
in \cite{KM2011}. 
(Only the Coulomb potential is treated explicitly in
\cite{KM2011,MatteStockmeyer2009a}; extensions to other potentials
are,
however, straightforward.)
\end{proof}

\smallskip

\noindent
What is crucial about the previous corollary is the range
of decay rates $\beta$ allowed for in \eqref{exp-NR} and
\eqref{exp-SR}.
For instance, suppose $V$ satisfies the conditions of
Theorem~\ref{thm-binding}(a)
in the SR case. Suppose further that $|e_{\sr}(V)|<1$,
which will hold true for weak potentials $V$ and is
known to be true in the Coulomb case, $V=V_\gamma$,
as long as the model is stable, i.e.
$\gamma\in(0,\nf{2}{\pi}]$; see \cite{RRSMS}.
Then the exponential decay rate for ground state eigenfunctions of
$h_\sr(V)$ is equal to $\beta_\el:=(1-(1+e_{\sr}(V))^2)^{\nf{1}{2}}$;
see \cite{CMS1990}.
Since $\Egs{\sr}(0)-\Egs{\sr}(V)>|e_{\sr}(V)|$ 
by Theorem~\ref{thm-binding}, there exist
$\beta\in(\beta_\el,1)$ such that $e^{\beta|\V{x}|}\Phi_\sr\in\HR$.

In the SR case the bounds of Theorem~\ref{thm-binding}
are consequences of certain bounds on
the {\em fiber} Hamiltonians
giving rise to some further interesting results.
In order to state the first one we introduce the functions
\begin{equation}\label{def-S}
T_\sr(\p):=(\p^2+1)^{\nf{1}{2}}-1\,,\quad
S(\p):= 1-(\p^2+1)^{-\nf{1}{2}},\qquad\p\in\R^3.
\end{equation}

\begin{theorem}[{\bf Upper bound on the mass shell}]\label{thm-ub}
Fix $d$, $g$, and $r$ in Hypothesis~\ref{hyp-G} and let
$\p_*\in\RR^3$.
Then we find
some $\gamma(\p_*)\equiv\gamma(\p_*,d,g,r)\in(0,1)$ such that, for all 
$(\omega,\vmu,\V{G})$ fulfilling Hypothesis~\ref{hyp-G} and $\p\in\RR^3$,
\begin{equation}\label{Diff4b}
\tfrac{1}{2}\big(\E{\sr}(\p_*+\p)+\E{\sr}(\p_*-\p)\big)
\le T_\sr(\p)+\E{\sr}(\p_*)-\gamma(\p_*)\,S(\p)\,.
\end{equation}
\end{theorem}

\begin{proof}
The assertion is proved in Lemma~\ref{LemSemRel}(b).
\end{proof}

\smallskip

\noindent
The previous theorem has the following immediate
corollary according to which the renormalized electron mass 
(i.e. the inverse of $(d^2/dt^2)\E{\sr}(t\,\V{u})|_{t=0}$) is always
strictly larger than its bare mass, which equals $1$ in the
units chosen in this paper.
The regularity assumptions on $\E{\sr}$ in the statement
can be verified, at least
for small coupling constants $|\ee|>0$ depending on $\Lambda$; 
see \cite[Theorem 7.1]{KM2012a}.
In this situation it is also known \cite{KM2012a} that
$\E{\sr}(\V{0})=\inf_{\R^3}\E{\sr}$.

\begin{corollary}[{\bf Renormalized electron mass}] \label{RenMass}
In the situation of Example~\ref{ex-Gphys}(ii)
let $|\ee|,\Lambda>0$ and let $\gamma(\V{0})$ be as in \eqref{Diff4b}.
If $\E{\sr}$ is twice continuously differentiable
in a neighborhood of zero, then
\begin{equation}\label{ren-mass}
\frac{d^2}{dt^2}\E{\sr}(t\,\V{u})\big|_{t=0}\le1-\gamma(\V{0})<1\,,
\qquad\V{u}\in\RR^3,\;|\V{u}|=1\,.
\end{equation}
\end{corollary}

\begin{proof}
By \eqref{Diff4b},
$\tfrac{1}{t^2}(\E{\sr}(t\,\V{u})+\E{\sr}(-t\,\V{u})-2\E{\sr}(t\,\V{u}))
\le \tfrac{2}{t^2}\{( t^2\V{u}^2+1)^{\nf{1}{2}}-1
-\gamma(\V{0}) \,S(t\,\V{u})\}$, $t>0$.
In the limit $t\downarrow0$ this gives \eqref{ren-mass},  if $\E{\sr}$ is
$C^2$ near $\V{0}$. 
\end{proof}

\smallskip

\noindent
As a final application we discuss the ultra-violet
behavior of the
ground state energy $\Egs{\sr}(0)$. 
In the rest of this section we only consider the situation of
Example~\ref{ex-Gphys}(ii).
To state and prove our corresponding results we introduce the bare mass
of the electron, $\m\ge0$, as an additional parameter
and display the UV cutoff parameter $\Lambda>0$ explicitly in the
notation.
More precisely, 
if we choose $\V{G}=\V{G}_\Lambda^\ee$ as in
\eqref{def-Gphys}, then we denote $w(\p)$ and $\ww$
as $w_\Lambda(\p)$ and $\ww_\Lambda$. For
all $\m\ge0$, $\Lambda>0$, and $\p\in\R^3$, we then define
\begin{align*}
\HamF{\sr,\Lambda,\m}(\p)
&:=\sqrt{w_\Lambda(\p)^2+\m^2}-\m+\Hf\quad
\textrm{with domain}\;\dom(\Hf)\,,
\\
\E{\sr,\Lambda,\m}(\p)&:=\inf\Spec(\HamF{\sr,\Lambda,\m}(\p))\,,
\\
\Ham{\sr,\Lambda,\m}(0)&:=\sqrt{\smash{\ww_\Lambda}^2+\m^2}-\m+\Hf\quad
\textrm{with domain}\;\dom((-\Delta_\x)^{\nf{1}{2}}+\Hf)\,,
\\
\Egs{\sr,\Lambda,\m}&:=\inf\Spec(\Ham{\sr,\Lambda,\m}(0))\,,
\end{align*}
so that
\begin{equation}\label{claudia0}
\Egs{\sr,\Lambda,\m}=\essinf{\p\in\R^3}\,\E{\sr,\Lambda,\m}(\p)\,,\qquad
\m\ge0\,,\;\Lambda>0\,.
\end{equation}
On account of
$$
0\le
(t^2+\m_1^2)^{\nf{1}{2}}-\m_1-(t^2+\m_2^2)^{\nf{1}{2}}+\m_2\le
\m_2-\m_1\,,
\quad 0\le\m_1\le\m_2\,,
$$
the difference of two Hamiltonians with different bare masses
extends to a bounded operator on the whole Hilbert space
with
$\|\HamF{\sr,\Lambda,\m_1}(\p)-\HamF{\sr,\Lambda,\m_2}(\p)\|
\le|\m_1-\m_2|$ and similarly for
$\Ham{\sr,\Lambda,\m}(0)$. In particular, all remarks on the
(essential) self-adjointness of the Hamiltonians with $\m=1$
are actually valid, for all $\m\ge0$. Furthermore,
\begin{equation}\label{claudia}
0\le\left\{
\begin{array}{c}
\E{\sr,\Lambda,\m_1}(\p)-\E{\sr,\Lambda,\m_2}(\p)
\\
\Egs{\sr,\Lambda,\m_1}-\Egs{\sr,\Lambda,\m_2}
\end{array}
\right\}
\le\m_2-\m_1\,,\quad0\le\m_1\le\m_2\,.
\end{equation}
In fact, every mass $\m>0$ is related to 
the bare mass one by scaling:
Let $(u\,\psi)(\k,\lambda)= \Lambda^{3/2}\,\psi(\Lambda\,\k,\lambda)$ be the 
dilatation on $L^2(\R^3\times \Z_2)$, and let $\Gamma(u)$ be the 
associated dilatation on the Fock space. 
The action of the unitary $\Gamma(u)$ is characterized by the
formulas
\begin{gather*}
\Gamma(u)\,a(f)=a(u\,f)\,\Gamma(u)\,,
\quad \Gamma(u)\,a^\dagger(f)=a^\dagger(u\,f)\,\Gamma(u),
\quad \Gamma(u)\,\Hf=\Lambda\,\Hf\,\Gamma(u)\,,
\\ 
\Gamma(u)\,\pf=\Lambda\,\pf\,\Gamma(u)\,,\quad\Gamma(u)\,\Omega=\Omega\,.
\end{gather*}
Moreover, $u\,\V{G}_\Lambda^\ee=\Lambda\,\V{G}_1^\ee$ by \eqref{def-Gphys}.
From these formulas we readily infer that
\begin{equation}\label{Intertwine}
\Gamma(u)\,\HamF{\sr,\Lambda,\m}(\p)=\Lambda\, 
\HamF{\sr,1,\m/\Lambda}(\p/ \Lambda)\, \Gamma(u)\,.
\end{equation}
In view of \eqref{claudia0} and \eqref{claudia} this permits to get
\begin{align}\label{claudia2}
\Lambda^{-1}\,\Egs{\sr,\Lambda,1}=\Egs{\sr,1,1/\Lambda}
\uparrow\Egs{\sr,1,0}\,,\quad\Lambda\uparrow\infty\,.
\end{align}
Our results imply that the limit in \eqref{claudia2}
is actually {\em non-zero}.

%
%%%%%%%%%%%%%%%%%%%%%%%%%%%%%%%%%%%%%%%%%%%%%%%%%%%%%%%%%
%%%%%%%%%%%%%%%%%%%%%%%%%%%%%%%%%%%%%%%%%%%%%%%%%%%%%%%%%
%
\begin{theorem}[{\bf UV-Asymptotics}]\label{thm-UV}
In the situation of Example~\ref{ex-Gphys}(ii) and with the notation
introduced above we have $\Egs{\sr,1,0}>0$. 
In particular, the leading asymptotics of
$\Egs{\sr,\Lambda,1}$ 
is linear in $\Lambda\to\infty$.
\end{theorem}
\smallskip

\noindent
An asymptotically linear growth of the self-energy 
has been observed earlier in \cite{LiebLoss2000}.
Notice that the 
self-energy grows at least as fast as
$\const\,\Lambda^{\nf{3}{2}}$ in the NR model; 
see \cite{LiebLoss2000}.

In view of the above simple remarks the
existence of a
non-vanishing linear contribution to $\Egs{\sr,\Lambda,1}$
is an immediate consequence of
the results of Section~\ref{Abb4}:

\begin{proof}
Applying successively
\eqref{claudia0}, \eqref{claudia}, and \eqref{joice} we get
$$
\Egs{\sr,1,0}=\tessinf{\R^3}\E{\sr,1,0}\ge\tessinf{\R^3}\E{\sr,1,1}
=\tessinf{|\p|\le\pmax_*}\E{\sr,1,1}(\p)\,.
$$
By \eqref{UniversalLowerBound} below the last essential infimum
is strictly positive.
\end{proof}

%%%%%%%%%%%%%%%%%%%%%%%%%%%%%%%%%%%%%%%%%%%%%%%%%%%%%%%%%%%%%%%%%%%%%%%%
%%%%%%%%%%%%%%%%%%%%%%%%%%%%%%%%%%%%%%%%%%%%%%%%%%%%%%%%%%%%%%%%%%%%%%%%
%%%%%%%%%%%%%%%%%%%%%%%%%%%%%%%%%%%%%%%%%%%%%%%%%%%%%%%%%%%%%%%%%%%%%%%%
%
%
%
%
\section{Main technical estimates}\label{Abb4}
\noindent
Before we derive our main theorems stated in Section~\ref{sec-results}
we develop the crucial technical ingredient underlying their proofs
in the present section. Throughout this section we shall always
assume that $(\omega,\vmu,\V{G})$ fulfill Hypothesis~\ref{hyp-G}.

In what follows we 
pick $\delta>0$, $e\ge0$, and define spectral projections
\begin{equation*}
\ol{\Pi}_\delta:=\id_{(e-\delta,e+\delta)}(\Hf)\,,\qquad
\Pi_\delta:=\id-\ol{\Pi}_\delta\,.
\end{equation*}
We shall use the following simple observation:
Recall the notation
$$
(a(k)\,\psi)^{(n)}(k_1,\ldots,k_n)
:=(n+1)^{\nf{1}{2}}\psi^{(n+1)}(k,k_1,\ldots,k_n)\,,
$$
almost everywhere, for $\psi=\{\psi^{(\ell)}\}_{\ell=0}^\infty\in\Fock_b$ and
$n\in\NN_0$, 
and $a(k)\,\Omega:=0$.
Suppose that $f_1,\ldots,f_n\in\mathfrak{h}$
have supports in $\{\omega\ge2\delta\}\times\ZZ_2$.
By the pull-through formula,
$a(k)\,F(\Hf)=F(\Hf+\omega(\V{k}))\,a(k)$,
we have, for every $\psi\in\Fock_b$,
\begin{align*}
&a(f_1)\ldots a(f_n)\,\ol{\Pi}_\delta\,\psi
\\
&=
\int \ol{f_1(k_1)}\dots \ol{f_n(k_n)}\,
\id_{(e-\delta,e+\delta)}\Big(\Hf+\sum_{j=1}^n\omega(\V{k}_j)\Big)\,
a(k_1)\ldots a(k_n)\,\psi\,dk_1\ldots dk_n\,,
\end{align*}
where $k_j=(\V{k}_j,\lambda_j)$, $j=1,\ldots,n$. 
Of course, if $\id_{(e-\delta,e+\delta)}(t+\Sigma)\not=0$ with
$\Sigma\ge2\delta$,
then $|t-e|\ge\Sigma-|t+\Sigma-e|\ge\delta$.
That is,
\begin{align*}
a(f_1)\ldots a(f_n)\,\ol{\Pi}_\delta
&=
\Pi_\delta\,a(f_1)\ldots a(f_n)\,\ol{\Pi}_\delta
\end{align*}
on $\Fock_b$, or,
\begin{align*}%\label{uwe0}
a(f_1)\ldots a(f_n)
&=
a(f_1)\ldots a(f_n)\,{\Pi}_\delta+
\Pi_\delta\,a(f_1)\ldots a(f_n)\,\ol{\Pi}_\delta
\end{align*}
on the domain of $\Hf^{\nf{n}{2}}$.
Combining this with \eqref{StandardEstimate}
we obtain in particular
\begin{align}\nonumber
|\sps{\phi}{&\,a(f_1)\ldots a(f_n)\,\psi}|
\\\label{uwe1}
&\le\Big(\prod_{j=1}^n\|f_j\|_{\nf{1}{2}}\Big)\,
\big(\,\|\phi\|\,\|\Hf^{\nf{n}{2}}\,{\Pi}_\delta\,\psi\|
+(e+\delta)^{\nf{n}{2}}\|\Pi_\delta\,\phi\|\,\|\psi\|\,\big)\,,
\end{align}
for all $\phi\in\Fock_b$ and $\psi\in\dom(\Hf^{\nf{n}{2}})$.
Let $d$, $g$, and $r$ be given by Hypothesis~\ref{hyp-G}.

\begin{lemma}\label{le-uwe}
With the assumptions and notation explained in the previous
paragraphs we find a universal constant, $\const>1$, such that,
for every normalized $\psi\in\dom(\Hf)$,
all $e\ge0$, $\delta>0$ with $r(2\delta)\le1$,
$\ve\in(0,1]$, and $\p\in\R^3$,
\begin{align}\nonumber
\|w(\p)\,\psi\|^2
&\ge g^2- r(2\delta)\,\big(2+9d\sps{\psi}{\Hf\,\psi}\big)
\\\label{uwe2}
&\quad
-\const\,d^2(1+e+\delta)(1+|\p|)
\big(\ve+\ve^{-1}\,\|(\Hf+1)\,\Pi_\delta\,\psi\|^2\big)
\,.
\end{align}
\end{lemma}

\begin{proof}
We set
$\V{G}_<:=\one_{\{\omega<2\delta\}}\,\G$, 
$\G_>:=\one_{\{ \omega\ge 2\delta\}}\,\G$, and
  $w_<(\p):=\vsigma\cdot(\p-\pf+ \vp(\V{G}_<))$,
so that
\begin{equation*}
w(\p)^2
=w_<(\p)^2
+ w_<(\p)\,(\vsigma\cdot\vp(\V{G}_>)) 
+(\vsigma\cdot\vp(\V{G}_>))\,w_<(\p)
+\vp(\V{G}_>)^2\quad\textrm{on}\;\sC.
\end{equation*}
Using 
$w_<(\p)^2\ge 0$, $a^{\dagger}(\G_>)\cdot
a(\G_>)\ge 0$, and the canonical commutation relations
we obtain by a straightforward computation, for $\psi\in\sC$,
$\|\psi\|=1$,
\begin{align*}
\sps{\psi}{w(\p)^2\,\psi}
&\ge 
4\Re \sps{(\p-\pf)\,\psi}{a(\G_>)\,\psi}
-2\Re\sps{\psi}{a(\vmu\cdot\V{G}_>)\,\psi}
\\
&\;+2\Re \sps{\psi}{i\vsigma\cdot a(\vmu\times \G_>)\,\psi}
+2 \Re \sps{\psi}{a(\G_>)\,a(\G_>)\,\psi}
\\
&\;
+4\Re \sps{a(\G_>)\,\psi}{ a(\G_<)\,\psi}
+4 \Re \sps{a^{\dagger}(\G_<)\, \psi}{a(\G_>)\,\psi}+\|\G_>\|^2.
\end{align*}
Next, we apply \eqref{StandardEstimate} to the two terms containing
$\V{G}_<$ in the last line,
\eqref{hyp-r} to the term $\|\G_>\|^2$, and
\eqref{eq-hyp-G}  and \eqref{uwe1}
to all remaining terms on the RHS of the previous estimate.
Proceeding in this way we arrive at
\begin{align*}
\|&w(\p)^2\,\psi\|^2
\\
&\ge g^2-r(2\delta)^2 
-4d\,\big(\|(\p-\pf)\,\psi\|\,\|\Hf^{\nf{1}{2}}\,\Pi_\delta\,\psi\|+
(e+\delta)^{\nf{1}{2}}\|(\p-\pf)\,\Pi_\delta\,\psi\|\big)
\\
&\quad-4d\,r(2\delta)\,\big(2\|\Hf^{\nf{1}{2}}\psi\|^2+\|\Hf^{\nf{1}{2}}\,\psi\|\big)
-2d^2\,\big(\|\Hf\,\Pi_\delta\,\psi\|+(e+\delta)\,\|\Pi_\delta\,\psi\|\big)
\\
&\quad-\const\,d\big(\|\Hf^{\nf{1}{2}}\,\Pi_\delta\,\psi\|+
(e+\delta)^{\nf{1}{2}}\|\Pi_\delta\,\psi\|\big)\,.
\end{align*}
Finally, we use 
$\|\pf\,\psi\|/d\le\|\Hf\,\psi\|\le\|\Hf\,\Pi_\delta\,\psi\|+e+\delta$
and
$\|\Hf^{\nf{1}{2}}\,\Pi_\delta\,\psi\|\le\|(\Hf+1)\,\Pi_\delta\,\psi\|^{\nf{1}{2}}$
as well as the following consequence of Young's inequality,
$t\,(1+t^{\nf{1}{2}})\le\ve+2t^2/\ve$,
$t\ge0$, $0<\ve\le1$,
to obtain \eqref{uwe2}, for all normalized $\psi\in\sC$.
An approximation argument extends it to all $\psi\in\dom(\Hf)$
of norm $1$.
\end{proof}

\smallskip

\noindent
In what follows we abbreviate
\begin{align*}
F_\sharp(\p):=\hat{\tau}_\sharp(\p)^2+w(\p)^2.
\end{align*}

\begin{lemma}\label{le-horst}
Let $\mathfrak{p}>0$. Then we
we find $\const_0>0$, depending only on $\mathfrak{p}$ and the
quantities $d$, $g$, and $r$ in Hypothesis~\ref{hyp-G},
such that, for all $|\p|\le\mathfrak{p}$, 
$\rho\in(0,1]$, and all normalized $\psi$ in the range of the spectral projection
$\id_{(-\infty,\,\E{\sharp}(\p)+\rho)}(\HamF{\sharp}(\p))$,
\begin{align}\label{horst1}
\sps{\psi}{&\,F_\sharp(\p)\,\psi}\ge\const_0-\rho^2.
\end{align}
\end{lemma}

\begin{proof}
We set $e:=\E{\sharp}(\p)$ and always assume that $\delta\in(0,1]$.
By \eqref{joice} we have $e\le\emax\equiv\emax(\pmax,d)$.
By assumption
$\|\hat{\tau}_{\sharp}(\p)\,\psi+(\Hf-e)\,\psi\|\le\rho$,
whence
\begin{align*}
2\rho^2+2\SP{\psi}{\hat{\tau}_{\sharp}(\p)^2\,\psi}
&\ge\|(\Hf-e)\,\psi\|^2\ge\|(\Hf-e)\,\Pi_\delta\,\psi\|^2.
\end{align*}
Moreover,
$\sps{\psi}{\Hf\,\psi}\le\sps{\psi}{\HamF{\sharp}(\p)\,\psi}\le e+\rho\le\emax+1$
and we observe that
$\|(\Hf+1)\,\Pi_\delta\,\psi\|\le(2+e)\|(\Hf-e)\,\Pi_\delta\,\psi\|/\delta$.
By virtue of Lemma~\ref{le-uwe} we deduce that the inequality
\begin{align*}
\SP{\psi}{&\,F_\sharp(\p)\,\psi}\ge[a-b\,x]_++x-\rho^2
\end{align*}
is satisfied at the point $x=\|(\Hf-e)\,\Pi_\delta\,\psi\|^2/2$,
where $[t]_+:=\max\{0,t\}$ and
\begin{align*}
a&:=g^2-r(2\delta)\,\big(2+9d(1+\emax)\big)-\const\,d^2(2+\emax)(1+|\p|)\,\ve\,,
\\
b&:=2\const\,d^2(2+\emax)^3(1+|\p|)/\delta^2\ve\,.
\end{align*}
Finally, we fix $\delta,\ve\in(0,1]$ such that
$a\ge g^2/2$ and observe that $b>1$ and, hence,
$
\inf_{x\ge0}\{[g^2/2-b\,x]_++x\}\ge g^2/2b
$.
\end{proof}

\smallskip

\noindent
In our applications we need bounds similar to \eqref{horst1} but
with $F_\sharp(\p)$ replaced by some other functions of $w(\p)$.
In order to derive them in Proposition~\ref{prop-hans} below we consider the spectral 
measures associated with $w(\p)^2$,
\begin{equation*}
d\mu_\psi(\lambda):=d\sps{\psi}{E_\lambda(w(\p)^2)\,\psi}\,,\qquad
\psi\in\CC^2\otimes\Fock_b\,,
\end{equation*}
and write $F_\sharp(\p)=f(w(\p)^2)$ so that
$f(t)=t^2/4+t$ in the NR case and $f(t)=2(t+1-\sqrt{t+1})$ in the SR case.
In the proof of Proposition~\ref{prop-hans} 
it is crucial that
$$
\int_{\{f\ge b\}}f\,d\mu_\psi\xrightarrow{\;b\to\infty\;}0\,,
$$
{\em uniformly} for all normalized $\psi$ in the range of
$\id_{(-\infty,\E{\sharp}(\p)+\rho)}(\HamF{\sharp}(\p))$.
To verify the uniformity of the above limit we shall apply (see
\eqref{hans1})
the following higher order estimate:

\begin{lemma}\label{le-hoe}
For every $\p\in\R^3$, it holds
$\dom(\HamF{\sharp}(\p)^{\nf{n}{2}})\subset\dom(\Hf^{\nf{n}{2}})$ and
\begin{equation*}%\label{hoe}
\big\|\Hf^{\nf{n}{2}}(\HamF{\sharp}(\p)+1)^{-\nf{n}{2}}\big\|
\le\const(d)\,,\qquad n\in\{1,\ldots,8\}\,.
\end{equation*}
\end{lemma}

\begin{proof}
According to Theorems~{4.2} and~{5.2} of \cite{Matte2009} and
\eqref{eq-hyp-G} we have
\begin{equation}\label{sergey1}
\big\|\Hf^{\nf{n}{2}}(\Ham{\sharp}(0)+1)^{-\nf{n}{2}}\big\|
\le\const(d)\,(\Egs{\sharp}(0)+1)^{2n-1}.
\end{equation}
(The same bound with a less explicit RHS has been obtained earlier
in \cite{FGS2001}.)
Since $\Ham{\sharp}(0)$ is unitarily equivalent to the direct integral
of the fiber Hamiltonians and since the corresponding unitary
transformation ($U_{\V{0}}$ in \eqref{def-Uq})
commutes with $\Hf$ the LHS of \eqref{sergey1} is equal to
$\sup_{\ve>0}\tesssup{\p\in\R^3}N_\ve(\p)$
with
$$
N_\ve(\p):=\|F_\ve\,(\HamF{\sharp}(\p)+1)^{-\nf{n}{2}}\|\,,
\;\;
\p\in\RR^3\,, \quad F_\ve:=\Hf^{\nf{n}{2}}(\id+\ve\,\Hf)^{-\nf{n}{2}}.
$$
Since $\p\mapsto(\HamF{\sharp}(\p)+1)^{-1}$ is norm-continuous
(see \cite{KM2012a} for the SR case) we know that
$N_\ve$ is continuous on $\R^3$ and, in particular, its essential supremum
is actually a supremum. Furthermore,
$\Egs{\sharp}(0)\le\const'(d)$ by Lemma~\ref{le-mass-shell}.
\end{proof}

\begin{proposition}\label{prop-hans}
Let $\mathfrak{p}>0$. Then there exist $\const_0,\const_1>0$,
depending only on $\mathfrak{p}$ and the quantities $d$, $g$,
and $r$ in Hypothesis~\ref{hyp-G},
such that, for all $|\p|\le\mathfrak{p}$, 
$\rho\in(0,1]$, and all normalized $\psi$ in the range of the spectral projection
$\id_{(-\infty,\,\E{\sharp}(\p)+\rho)}(\HamF{\sharp}(\p))$,
\begin{equation}\label{hans}
\mu_\psi\big([a,b]\big)
\ge(\const_0-\const_1/b_0-2a-\rho^2)/b_0\,,
\quad0<a<b_0\le b\,,\;a\le1\,.
\end{equation}
\end{proposition}

\begin{proof}
On account of Lemma~\ref{le-hoe} we have,
for every $b>0$,
\begin{align}\nonumber
\int_{\{f\ge b\}} f\,d\mu_\psi
&\le b^{-1}\int_\RR f^{2}\,d\mu_\psi
\le\const\,b^{-1}\,\|(w(\p)^4+i)\,\psi\|^2
\\\label{hans1}
&\le\const'(\mathfrak{p},d)\,b^{-1}\,\|(\Hf+1)^4\,\psi\|^2
\le\const_1(\mathfrak{p},d)\,b^{-1}.
\end{align}
Dropping the argument $(\mathfrak{p},d)$ of $\const_1$ 
we infer from \eqref{horst1} that, for all $0<a<b$,
$$
\const_0-\rho^2\le\int_\RR f\,d\mu_\psi\le a+b\,\mu_\psi(\{a\le f\le b\})+\const_1/b\,,
$$
that is, 
$$
\mu_\psi(\{2a\le f\le b\})\ge(\const_0-\const_1/b-2a-\rho^2)/b\,.
$$
Recall that $\mu_\psi((-\infty,0))=0$.
In the NR case $t\le f(t)\le t\,(t+1)$ and in the
SR case $t\le f(t)\le2t$ and we readily conclude the proof.
\end{proof}

\begin{corollary}\label{cor-ULB}
Let $\mathfrak{p}>0$. Then there exists $\const>0$,
depending only on $\mathfrak{p}$ and the quantities $d$, $g$,
and $r$ in Hypothesis~\ref{hyp-G},
such that, for all $|\p|\le\mathfrak{p}$, 
\begin{equation}\label{UniversalLowerBound}
\E{\sharp}(\p)\ge\const\,,
\end{equation}
and, with $\cM_\rho$ denoting the set of all $\psi$ as in the
statement of Proposition~\ref{prop-hans},
\begin{align}\label{hans3}
\liminf_{\rho\downarrow0}\inf_{\psi\in\cM_\rho}\sps{\psi}{\hat{\tau}_\sharp(\p)\,(\hat{\tau}_\sharp(\p)+1)^{-1}\,\psi}
\ge\const\,.
\end{align}
\end{corollary}

\begin{proof}
Fix $a,b,\rho>0$ such that the RHS in \eqref{hans} is $\ge \const_0/2b$
 and set $g(t):=t/2$ in the NR case
and $g(t)=\sqrt{t+1}-1$ in the SR case.
Then
\begin{align*}
\E{\sharp}(\p)&=\inf_{\psi\in\cM_\rho}\sps{\psi}{\HamF{\sharp}(\p)\,\psi}
\ge\inf_{\psi\in\cM_\rho}\sps{\psi}{\hat{\tau}_\sharp(\p)\,\psi}
\\
&=\inf_{\psi\in\cM_\rho}\int_\RR g(t)\,d\mu_\psi(t)
\ge g(a)\inf_{\psi\in\cM_\rho}\mu_\psi([a,b])\ge \const_0\,g(a)/2b\,.
\end{align*}
The same argument with $g$ replaced by $g/(g+1)$ gives \eqref{hans3}.
\end{proof}

%%%%%%%%%%%%%%%%%%%%%%%%%%%%%%%%%%%%%%%%%%%%%%%%%%%%%%%%%%%%%%
%%%%%%%%%%%%%%%%%%%%%%%%%%%%%%%%%%%%%%%%%%%%%%%%%%%%%%%%%%%%%%
%%%%%%%%%%%%%%%%%%%%%%%%%%%%%%%%%%%%%%%%%%%%%%%%%%%%%%%%%%%%%%
%
%
%
\section{The non-relativistic case}\label{SecNR}

\noindent
In this section we prove Theorem~\ref{thm-binding} in the NR case, 
$\sharp=\nr$, by a variational argument employing a
trial function (see \eqref{ida5}) resembling the one in \cite{ChenVV}.
Thanks to the results of Section~\ref{Abb4} 
we may construct this trial function by means of
minimizing sequences for certain fiber Hamiltonians $\HamF{\nr}(\V{q})$;
we do not assume existence of minimizers of 
$\HamF{\nr}(\V{0})$ as in \cite{ChenVV}.
Some further modifications (the unitaries $U_{\q}$ and $U_R$ below) allow us
to drop the radial symmetry of the external potential assumed earlier.
(This restriction has been overcome in \cite{BLV2005}, too.)
By the use of $U_{\q}$ and $U_R$ it is also immaterial whether
$E_\nr$ attains its minimum at zero or not.
Finally, we remark that the momentum cut-off $\chi_\vr$
in Part~(b) of the proof is inserted in order to handle
all $V_-\in L^{\nf{3}{2}}(\R^3)$.
The crucial point is that we cannot expect the second derivatives of
the eigenfunctions $\psi_\lambda$ of $h_\nr(V_\lambda)$, $\lambda>1$,
to belong to $L^2(\R^3)$,
if the singularities of $V_-$ are not square-integrable.
Therefore, we have to regularize the expression
$\partial_{x_\nu}\psi_\lambda$ in momentum space before we take
further derivatives of it in some of the estimates below.
What is important to observe the enhanced binding is that the norm of
$\chi_\vr\,\partial_{x_\nu}\psi_\lambda$ does not vanish in the
limit $\lambda\downarrow1$ after suitable normalization of
$\psi_\lambda$.
This follows from Theorem~\ref{zero-resonance} whose proof is
inspired by \cite{SorStock}.

\smallskip

\begin{proof}[Proof of Theorem~\ref{thm-binding}: the NR case]
Let $\pmax_*$ be as in Lemma~\ref{le-mass-shell} and
$\p_*\in \R^3$, $|\p_*|\le\pmax_*$. 
Let $R\in\mathrm{SO}(3,\RR)$ be some rotation matrix
to be specified later on.
Setting $\q:=R\,\p_*$ and $\v{\p}:=\p-\pf+\vp(\V{G})$, $\p\in\R^3$,
we have the following identity for the fiber Hamiltonians,
\begin{equation}\label{ida1}
\HamF{\nr}(\q+\p)=\HamF{\nr}(\q)+  \p\cdot \v{\q}
+ \tfrac{\p^2}{2}\,,\qquad\p\in\R^3,
\end{equation}
in the sense of quadratic forms on, e.g., the domain
$\dom(\Hf)\supset\dom(\HamF{\nr}(\p))$, $\p\in\R^3$.
Since we assume that $(\omega,\vmu,\V{G})$ have the special form
of Example~\ref{ex-Gphys}(i) there is a unitary operator,
$U_R\in\cB(\CC^2\otimes\Fock_b)$,
depending on $R$ and the choice of the polarization vectors,
such that 
\begin{align}\label{ida2}
U_R\,\HamF{\nr}(\q)\,U_R^*&=\HamF{\nr}(\p_*)\,,
\qquad
U_R\,\p\cdot\v{\q}\,U_R^*
=\p\cdot(R\,\v{\p_*})\,;
\end{align}
see, e.g., \cite[Lemma~2.10]{Hiroshima2007} for details.
Next, we pick normalized 
$$
\phi_j\in\Ran\big(\id_{(-\infty,\E{\nr}(\p_*)+1/j)}(\HamF{\nr}(\p_*))\big)\,,\qquad
j\in\NN\,,
$$
so that $\phi_{j}\in\dom(\Hf^2)$ by Lemma~\ref{le-hoe}.
Furthermore, we pick normalized $\psi_1,\psi_2\in C_0^\infty(\R^3)$ 
satisfying $\psi_1=\bar{\psi}_1$, $\psi_2=-\bar{\psi}_2$, and let
$\{\phi^\prime_{j}\}_j$ denote another sequence of normalized vectors in 
$\dom(\Hf^2)\subset\CC^2\otimes\Fock_b$ to be specified later on
satisfying 
\begin{equation}\label{ida4}
\sps{\phi_j}{\phi_j^\prime}\in\R\,.
\end{equation} 
For $\eta <0$, we finally define unnormalized trial vectors 
$\psi_{\textrm{tr},j}$ by
\begin{equation}\label{ida5}
\psi_{\textrm{tr},j}:=U_{\q}^*\,\cF^*\hat{\psi}_{\textrm{tr},j}\,,\qquad
\hat{\psi}_{\textrm{tr},j}(\p):=
\hat{\psi}_1(\p)\,U_R^*\,\phi_{j}
+\eta\,\hat{\psi}_2(\p)\,U_R^*\,\phi^\prime_{j}\,,
\end{equation}
where $U_\q$ is the unitary operator defined in \eqref{def-Uq}.
By definition,
\begin{align}\label{ida6}
\Re\sps{\psi_1}{\psi_2}&=0\,,\qquad 
\sps{\psi_1}{\nabla_\x \psi_1}=\sps{\psi_2}{\nabla_\x\psi_2}=0\,,
\\\label{ida7}
\Re\sps{\psi_1}{h_{\nr}(V)\,\psi_2}&=0\,,\qquad
\Im\sps{\psi_1}{-i\nabla_\x\,\psi_2}=0\,,
\\\label{ida8}
\|\psi_{\textrm{tr},j}\|^2&=1+\eta^2.
\end{align}
In view of \eqref{fib-dec} and \eqref{ida8} we have
\begin{align*}
(1+\eta^2)\,\Egs{\nr}(V)&\le\SP{\psi_{\textrm{tr},j}}{\Ham{\nr}(V)\,\psi_{\textrm{tr},j}}
\\
&=\int_{\R^3}\sps{\hat{\psi}_{\textrm{tr},j}(\p)}{
\HamF{\nr}(\q+\p)
\,\hat{\psi}_{\textrm{tr},j}(\p)}_{\CC^2\otimes\Fock_b}\,d^3\p
+\sps{\psi_{\textrm{tr},j}}{V\,\psi_{\textrm{tr},j}}\,.
\end{align*}
Employing \eqref{ida1}--\eqref{ida7} we find after some easy computations
\begin{align}\nonumber
\Egs{\nr}(V)-(c_V+\const\,d^2)\,\eta^2
&\le
\sps{\psi_1}{h_{\nr}(V)\,\psi_1}+ \sps{\phi_j}{ \HamF{\nr}(\p_*)\,\phi_j}
\\ 
\nonumber
&\quad+ \eta^2 \,\sps{\psi_2}{ h_{\nr}(V)\,\psi_2}
+\eta^2\,\sps{\phi_j'}{ \HamF{\nr}(\p_*)\,\phi_j'}
\\ 
\nonumber
&\quad+ 2\eta\, 
\Re\big\{\sps{\psi_1}{\psi_2}\,\sps{\phi_j}{\HamF{\nr}(\p_*)\,\phi_j'} 
\big\}
\\ \label{EnergyEst1}
&\quad+ 2 \eta \,  \sps{\psi_1}{-\imath \nabla_\x \psi_2}\cdot \Re
\sps{R\,\v{\p_*}\,\phi_j}{\phi_j'}\,.
\end{align}
In the first line we also applied the lower bound \eqref{lb-Ham}.
By virtue of Corollary~\ref{cor-ULB} we find some $\const_0>0$,
depending only on $\pmax_*\equiv\pmax_*(d)$ and the quantities
$d$, $g$, and $r$ in Hypothesis~\ref{hyp-G},
such that
\begin{equation}\label{eq-D1}
(3/2)\,\liminf_{j\to\infty} \sps{\phi_{j}}{ \v{\p_*}^2 \phi_{j}}\ge
\liminf_{j\to\infty} \sps{\phi_{j}}{ \hat{\tau}_\nr(\p_*)\,\phi_{j}}\ge\const_0\,.
\end{equation}
By the higher order estimates of Lemma~\ref{le-hoe}
we know that $\sup_j\|v_{\nu}(\p_*)\,\phi_{j}\|$ is finite,
where $v_{\nu}(\p_*)$ is the $\nu$-th component of $\v{\p_*}$.
Passing to a suitable subsequence, if necessary, we may thus define
\begin{equation*}
c_1(\nu) := \lim_{j\to\infty}\|v_{\nu}(\p_*)\,\phi_{j}\|\,,
\quad\nu=1,2,3\,,\qquad
c_1^2:=\frac{1}{3}\sum_{\nu=1}^3c_1(\nu)^2\ge2\const_0/9\,,
\end{equation*}
For $\nu_0\in\{1,2,3\}$ with $c_1(\nu_0)^2\ge c_1^2$, we set
\begin{equation*}
\phi'_{j}:= v_{\nu_0}(\p_*)\,\phi_{j} 
\cdot \|v_{\nu_0}(\p_*) \,\phi_{j}\|^{-1}.
\end{equation*}
This choice is allowed since $\sps{\phi_{j}}{\phi'_{j}}\in \R$ due to
the fact that $v_{\nu_0}(\p_*)$ is symmetric and $\phi'_{j}\in \dom(\Hf^2)$ by
Lemma~\ref{le-hoe} and a straightforward calculation.
Furthermore, 
$\sps{\phi'_{j}}{\HamF{\nr}(\p_*)\, \phi'_{j}}
\le
\const'(\pmax_*,d)\,\|(\Hf+1)^2\,\phi_j\|^2/c_1(\nu_0)$
and the higher order estimates ensure the
existence
of some $c_2>0$, depending only on $\pmax_*$ and $d$, such that
$\sps{\phi'_{j}}{\HamF{\nr}(\p_*)\, \phi'_{j}}\le c_2$, for all $j$.
Since also 
$\sup_j|\sps{v_\nu(\p_*)\,\phi_{j}}{\phi'_{j}}|\le
\const(\pmax_*,d)\,\sup_j\|(\Hf+1)\,\phi_j\|^2<\infty$
we may define, at least along some suitable subsequence, 
\begin{align*}
\valpha&:=
\lim_{j\to\infty}\Re\sps{\v{\p_*}\,\phi_{j}}{\phi'_{j}}\,,\qquad
\textrm{so that}\qquad
|\valpha |\ge c_1\,.
\end{align*}
In fact,
the $\nu_0$-component of $\valpha$ is just equal to
$c_1(\nu_0)$.
We are still free to choose the rotation $R$ in \eqref{EnergyEst1}. We set
$\vbeta:=\sps{\psi_1}{-\imath \nabla_\x\psi_2}\in\RR^3$ and
choose it such that 
$\vbeta\cdot(R\,\valpha)=|\valpha|\,|\vbeta|\ge c_1\,|\vbeta|$.
Plugging the new notation into \eqref{EnergyEst1}, 
passing to the limit $j\to\infty$, and taking also
$$
\lim_{j\to\infty}\sps{\phi_j}{\HamF{\nr}(\p_*)\,\phi_j'}
=\E{\nr}(\p_*)\lim_{j\to\infty}\sps{\phi_j}{\phi_j'}\in\R
$$
into account, we readily arrive at
\begin{align}\nonumber
\Egs{\nr}&(V)-\E{\nr}(\p_*)-\sps{\psi_1}{ h_{\nr}(V)\,\psi_1}
\\ \label{eq-D4}
&\le
\eta^2 \, \big\{\sps{\psi_2}{h_{\nr}(V)\,\psi_2}+c_2+c_{\wt{V}}+\const\,d^2\big\}
+ 2\eta\,c_1\,|\vbeta|\,.
\end{align}
Here $\wt{V}$ is any potential satisfying the condition
in \eqref{def-cV}, i.e. $c_{\wt{V}}<\infty$, and $\wt{V}\le V$.
In particular, the curly bracket $\{\cdots\}$ in \eqref{eq-D4}
is strictly positive. Minimizing the RHS with respect to $\eta<0$ 
and applying \eqref{joice} on the LHS
we thus
obtain
\begin{align}
\Egs{\nr}&(V)-\Egs{\nr}(0)-\sps{\psi_1}{ h_{\nr}(V)\,\psi_1}\label{eq-D44}
\le
\frac{-|\vbeta|^2/\const_1}{\sps{\psi_2}{h_{\nr}(V)\,\psi_2}+c_{\wt{V}}+\const_1}\,,
\end{align}
where $\const_1>0$ depends only on $\pmax_*=\pmax_*(d)$, $d$, $g$, and $r$
and $\vbeta=\sps{\psi_1}{-i\nabla_\x\psi_2}$.
Since $C_0^\infty(\R^3)$ is a form core for $h_\nr(V)$ 
an approximation argument shows that
\eqref{eq-D44} is actually valid, for every real-valued normalized 
$\psi_1\in\form(h_\nr(V))$.
If $\form(h_\nr(V))\subset H^1(\R^3)$, then \eqref{eq-D44}
applies to every purely imaginary normalized $\psi_2\in\form(h_\nr(V))$.

\smallskip

\noindent
(a):
Under the assumptions of Theorem~\ref{thm-binding}(a)
the electron operator
$h_{\nr}(V)$ has a normalized, real-valued ground state 
eigenfunction, $\psi$. We set $\psi_1:=\psi$.
Since the distributional Laplacian $\Delta_\x\psi\in\sD'(\R^3)$
is non-zero we find some $\phi\in{C}^\infty_0(\R^3,\R)$ 
such that $\sps{\psi}{-\Delta_\x\phi}>0$.
We choose $\mu\in \{1,2,3\}$ such that 
$\sps{\psi}{-\partial_{x_\mu}^2\phi}>0$ and
set $\psi_2:=-i\partial_{x_\mu}\phi/\|\partial_{x_\mu}\phi\|$. 
Then
$\vbeta\not=0$
and the assertion follows from \eqref{eq-D44}.

\smallskip

\noindent
(b):
We replace $V$ by $V_\mu$, for $\mu\in(0,1)$.
Since $\lambda=1$ is the coupling constant threshold
there is a normalized, strictly positive ground state 
eigenfunction $\psi_\lambda>0$ of $h_{\nr}(V_\lambda)$,
for all $\lambda>1$, i.e.
$h_{\nr}(V_\lambda)\,\psi_\lambda=e_\lambda\, \psi_\lambda$ with $e_\lambda<0$.
According to Theorem~\ref{zero-resonance}
we find a sequence $\lambda_j\downarrow1$
such that the vectors
$(-\Delta)^{\nf{1}{2}}\psi_{\lambda_j}/\|(\lambda_j\,V_-)^{\nf{1}{2}}\psi_{\lambda_j}\|$
converge to some non-zero limit.
Passing to a subsequence, if necessary, we may assume that
$\|(-\Delta)^{\nf{1}{2}}\psi_{\lambda_j}\|
\le 3^{\nf{1}{2}} \|\partial_{x_{\nu} }\psi_{\lambda_j}\|$,
for some fixed $\nu\in\{1,2,3\}$ and all $j$.
Then the vectors
$\partial_{x_\nu}\psi_{\lambda_j}/\|(\lambda_j\,V_-)^{\nf{1}{2}}\psi_{\lambda_j}\|$
also have a non-zero limit.
Let $\chi_{\vr}:=\id_{(-\Delta)^{\nf{1}{2}}\le\vr}$
be a cut-off in momentum space.
Then we find $\alpha,\vr>0$ such that
$\|\chi_\vr\,\partial_{x_\nu}\psi_{\lambda_j}\|\ge\alpha\,
\|(\lambda_j\,V_-)^{\nf{1}{2}}\psi_{\lambda_j}\|$, 
for large $j$.
Now, we choose 
$\psi_1:= \psi_{\lambda_j}$ and
$\psi_2:=-i\chi_\vr\,\partial_{x_{\nu}}\psi_{\lambda_j}/
\|\chi_\vr\,\partial_{x_{\nu}}\psi_{\lambda_j}\|$.
Notice that $\psi_2$ is purely imaginary
because $\psi_{\lambda_j}$ is real-valued and the
cut-off $\chi_\vr$ is symmetric about the origin in 
momentum space.

Since $V_\pm\in (L^{\nf{3}{2}}+L^\infty)(\RR^3)$ we know that
$\psi_1\in H^1(\RR^3)$ and we may write $\vbeta$ as
$\vbeta=\sps{-i\nabla_\x\psi_{\lambda_j}}{-i\chi_\vr\,\partial_{x_{\nu}}\psi_{\lambda_j}}/\|\chi_\vr\,\partial_{x_{\nu}}\psi_{\lambda_j}\|$,
which shows that
$$
| \vbeta|\ge\|\chi_\vr\,\partial_{x_{\nu}}\psi_{\lambda_j}\|\ge\alpha\,
\|(\lambda_j\,V_-)^{\nf{1}{2}}\psi_{\lambda_j}\|\,.
$$
Furthermore, we choose
$\wt{V}:=-2V_-$ in \eqref{eq-D44}
for all $\lambda\le2$. Applying
\eqref{eq-D44} with $V=V_\mu$, taking the above remarks into account, and using 
$h_{\nr}(V_\mu)=h_{\nr}(V_{\lambda})+(\lambda-\mu)\,V_-$ 
we arrive at
\begin{align}\nonumber
\Egs{\nr}(V_\mu )&-\Egs{\nr}(0)
-e_{\lambda_j}-(\lambda_j-\mu)\sps{\psi_{\lambda_j}}{ V_-\psi_{\lambda_j}}
%\\ \label{EnergyEst3}&
\le \frac{-\alpha^2\,\|(\lambda_j\,V_-)^{\nf{1}{2}}\psi_{\lambda_j}\|^2/\const_1}{
\sps{\psi_2}{h_{\nr}(V_\mu)\,\psi_2}+c_{-2V_-}+\const_1},
\end{align}
for $0<\mu<1<\lambda_j\le2$, where
\begin{align*}
\sps{\psi_2}{
h_{\nr}(V_\mu)\,\psi_2}
&\le\vr^2/2+\|V_{+,2}\|_\infty+\const\,\|V_{+,1}\|_{\nf{3}{2}}\,\vr^2
=:\const(V_\pm)
\,,
\end{align*} 
since $\supp(\wh{\psi}_2)\subset\{|\p|\le\vr\}$
and $\sps{\psi_2}{V_{+,1}\psi_2}\le\const\,\|V_{+,1}\|_{\nf{3}{2}}^2
\,\|\nabla\psi_2\|^2$ by H\"older's and Sobolev's inequalities.
Hence,
\begin{align*}
\frac{\Egs{\nr}(V_\mu )-\Egs{\nr}(0)}{\|(\lambda_j\,V_-)^{\nf{1}{2}}\psi_{\lambda_j}\|^2}
\,-\,\frac{\lambda_j-\mu}{\lambda_j}\,
%\\ \label{EnergyEst3}&
\le \frac{-\alpha^2/\const_1}{\const(V_\pm)+c_{-2V_-}+\const_1}=:-\const_\star
\,.
\end{align*}
Now, fix $\delta>0$ such that
$2\delta=\const_\star\equiv\const_\star(d,g,r,V_\pm)$,
and then fix $j_0$ such that 
$(\lambda_{j_0}-1+\delta)/\lambda_{j_0}-\const_\star\le-\const_\star/4$.
Then $\|(\lambda_{j_0}\,V_-)^{\nf{1}{2}}\psi_{\lambda_{j_0}}\|$
is some $(d,g,r,V_\pm)$-dependent constant and we conclude.
\end{proof}
%
%%%%%%%%%%%%%%%%%%%%%%%%%%%%%%%%%%%%%%%%%%%%%%%%%%%%%%%%%%%%
%%%%%%%%%%%%%%%%%%%%%%%%%%%%%%%%%%%%%%%%%%%%%%%%%%%%%%%%%%%%
%
\section{The semi-relativistic case}\label{SecSR}

\noindent
In this section we prove the statements of 
Theorem~\ref{thm-binding} in the SR case.
We also derive the bound \eqref{Diff4b} asserted in
Theorem~\ref{thm-ub}.

Recall the definitions of $S$ and $T_\sr$ in \eqref{def-S}
and $\v{\p}=\p-\pf+\vp(\V{G})$.
We have
$S=T_\sr/(T_\sr+1)$.
Given $\p,\p_*\in\R^3$, we shall use
the following notation for resolvents, where $\eta>0$,
\begin{align*}
R_1(\eta)&:=(w(\p_*)^2+\p^2+1+\eta)^{-1},
\quad
R_2(\eta):= (w(\p_*+\p)^2+1+\eta)^{-1}.
\end{align*}

\begin{lemma}\label{LemSemRel}
(a) For all $\phi\in\dom(\Hf)$, $\|\phi\|=1$, 
and $\p,\p_*\in\R^3$, we have
\begin{align}\nonumber
\sps{\phi}{\HamF{\sr}(\p_*+\p)\,\phi}
&\le \sps{\phi}{\HamF{\sr}(\p_*)\,\phi}
-S(\p)\,\sps{\phi}{\hat\tau_\sr(\p_*)\,(\hat\tau_\sr(\p_*)+1)^{-1}\,\phi}
\\\label{Diff2} 
&\quad+T_\sr(\p)+\int_0^\infty 2 \sps{R_1(\eta)\,\phi}{\p\cdot\v{\p_*}\,
R_1(\eta)\,\phi}\,\frac{\eta^{\nf{1}{2}}d\eta}{\pi}\,.
\end{align}
(b) The bound \eqref{Diff4b} holds true.
\end{lemma}
%
%%%%%%%%%%%%%%%%%%%%%%%%%%%%%%%%%%%%%%%%%%%%%%%%%%%%%%%%%
%
\begin{proof}
The following proof of (a) is a strengthened version
of an argument used to derive a non-strict inequality
on the binding energy in \cite{KMS2009a}.

As a consequence of Lemma~A.1 of \cite{KM2012a}
both resolvents $R_1(\eta)$ and
$R_2(\eta)$ map $\dom(\Hf^\nu)$ into itself, for every $\nu\ge1/2$.
Taking this into account and writing
$w(\p_*+\p)^2\,\phi= 
(w(\p_*)^2+\p^2)\,\phi+2 \p\cdot\v{\p_*}\,\phi$,
for $\phi\in\dom(\Hf^2)$, we readily obtain
\begin{align*}
\sps{\phi}{R_1(\eta)\,\phi} &
= \sps{\phi}{R_2(\eta)\,\phi}+ 
2 \sps{R_2(\eta)\,\phi}{\p\cdot\v{\p_*}\,R_1(\eta)\,\phi}
\\
&= \sps{\phi}{R_2(\eta)\,\phi}+ 
2 \sps{R_1(\eta)\,\phi}{\p\cdot\v{\p_*}\,R_1(\eta)\,\phi}
\\
&\quad     
-4 \sps{R_2(\eta)\,\p\cdot\v{\p_*}\,R_1(\eta)\,\phi }{\p\cdot\v{\p_*}\,R_1(\eta)\,\phi}\,.
\end{align*}
Note that the expression in the last line is negative
since $R_2(\eta)$ is positive.
Dropping this term and
using the formula 
$$
A^{\nf{1}{2}}\phi=
\int_0^\infty (1-\eta
(A+\eta)^{-1})\,\phi\,\frac{d\eta}{\pi\,\eta^{\nf{1}{2}}}\,,
$$
valid for any
positive operator $A$ in some Hilbert space and $\phi\in \dom(A)$, we
obtain,
for normalized $\phi$,
\begin{align}\nonumber
\SP{\phi}{\hat\tau_\sr(\p_*+\p)\,\phi}
&\le\SP{\phi}{(w(\p_*)^2+\p^2+1)^{\nf{1}{2}}\,\phi}-1
\\ \label{Diff1}
& \quad+ \int_0^\infty 2 \sps{R_1(\eta)\,\phi}{\p\cdot\v{\p_*}
\,R_1(\eta)\,\phi}\,\frac{\eta^{\nf{1}{2}}\,d\eta}{\pi}\,,
\end{align}
which makes sense since 
$\|\Hf\,R_1(\eta)\,(\Hf+1)^{-1}\|\le\const(d)\,(1+\eta)^{-1}$
and, hence,
\begin{equation}\label{juergen}
\|\v{\p_*}\,R_1(\eta)\,\psi\|\le 
\const(\p_*,d)\,(1+\eta)^{-1}\,\| (\Hf+1)\,\psi\|\,,
\quad\psi\in\dom(\Hf)\,,
\end{equation} 
by Lemma~A.1 in \cite{KM2012a}.
Next, we observe that
$$
w(\p_*)^2+\p^2+1=a^2-2b\,,\quad
a:=\hat\tau_\sr(\p_*)+T_\sr(\p)+1\,,\;\;
b:=\hat\tau_\sr(\p_*)\,T_\sr(\p)\,.
$$
In a spectral representation of $w(\p_*)$ we may now apply the
inequality between
geometric and arithmetic means, 
$\sqrt{a(a-2b/a)}\le a-b/a$,
to see that the terms in the first line of the RHS of \eqref{Diff1}, 
where $\|\phi\|=1$,
are not greater than
$$
\sps{\phi}{\hat\tau_\sr(\p_*)\,\phi}+T_\sr(\p)
-T_\sr(\p)\,\SP{\phi}{\hat\tau_\sr(\p_*)\,(\hat\tau_\sr(\p_*)+T_\sr(\p)+1)^{-1}\,\phi}\,,
$$
where
$(\hat\tau_\sr(\p_*)+T_\sr(\p)+1)^{-1}\ge(T_\sr(\p)+1)^{-1}(\hat\tau_\sr(\p_*)+1)^{-1}$.
Adding $\sps{\phi}{\Hf\,\phi}$ on both sides of \eqref{Diff1} and
employing these bounds we arrive at \eqref{Diff2} with
$\phi\in\dom(\Hf^2)$.
Since, for every $\q\in\R^3$, we know that
$\dom(\HamF{\sr}(\q))=\dom(\Hf)$
and the graph norms of $\HamF{\sr}(\q)$ and $\Hf$ are equivalent
\cite{KM2012a} we obtain \eqref{Diff2}
with $\phi\in\dom(\Hf)$ by an approximation argument using \eqref{juergen}.

\smallskip

\noindent
(b):
The integral in the second line of \eqref{Diff2} is an odd function of $\p$
and cancels out when we add a copy of \eqref{Diff2} with
$\p$ replaced by $-\p$ to it. 
Therefore, \eqref{Diff4b} follows from \eqref{Diff2} 
upon using $\E{\sr}(\p_*\pm\p)\le \HamF{\sr}(\p_*\pm\p)$, inserting
normalized vectors $\phi_j$ in the range of the spectral projection
of $\HamF{\sr}(\p_*)$ corresponding to the interval
$(-\infty,\E{\sr}(\p_*)+1/j]$ and applying \eqref{hans3}.
%(Here we again use $\dom(\HamF{\sr}(\q))=\dom(\Hf)$.)
\end{proof}

\smallskip

\begin{proof}[Proof of Theorem~\ref{thm-binding}: The SR case]
Let $\pmax_*>0$ be the parameter appearing in
Lemma~\ref{le-mass-shell} and
set $\q:= \p_*$ in \eqref{fib-dec}, where $|\p_*|\le\pmax_*$. 
We apply \eqref{Diff2} to estimate the 
expectation of 
$$
\cF\,U_{\p_*}\,{\mathbbm{H}}_{\sr}(V)\,U_{\p_*}^*\cF^*
=\int_{\RR^3}^\oplus\HamF{\sr}(\p_*+\p)\,d^3\p+\cF\,V\,\cF^*
$$ 
in a trial vector
$\hat{\psi}_{\textrm{tr}}(\p)=\hat{\psi}_1(\p)\,\phi$
with $\phi\in \dom(\Hf^2)$, $\|\phi\|=1$, and
$\psi_1\in C_0^\infty(\R^3)$, 
$\bar{\psi}_1=\psi_1$, i.e. $|\hat{\psi}_1(\p)|=|\hat{\psi}_1(-\p)|$.
Since the integral in the second line of 
\eqref{Diff2} is an odd function of $\p$
it drops out when we integrate with respect to the symmetric
measure $|\hat{\psi}_1(\p)|^2d^3\p$ and we arrive at
\begin{align}\nonumber
\Egs{\sr}(V)\,\|\psi_1\|^2
&\le \int_{\RR^3} 
|\hat{\psi}_1(\p)|^2\SP{\phi}{\HamF{\sr}(\p_*+\p)\,\phi}\,d^3\p
+\sps{\psi_1}{V\,\psi_1}
\\
\nonumber
&\le \sps{\psi_1}{h_{\sr}(V)\,\psi_1}
+\|\psi_1\|^2\,\sps{\phi}{\HamF{\sr}(\p_*)\,\phi}
\\
&\qquad\label{SR-EnergyEst}
- \int_{\RR^3}|\hat{\psi}_1(\p)|^2\,S(\p)\, d^3\p\,
\SP{\phi}{\hat\tau_\sr(\p_*)(\hat\tau_\sr(\p_*)+1)^{-1}\phi}\,.
\end{align}
Let $\phi_j$ be as in the proof of Lemma~\ref{LemSemRel}(b)
so that  
$\sps{\phi_j}{\HamF{\sr}(\p_*)\,\phi_j}\to\E{\sr}(\p_*)$.
Substituting $\phi_j$ for $\phi$ in \eqref{SR-EnergyEst},
passing to the limit $j\to \infty$,
and taking \eqref{hans3} into account
we deduce that
\begin{equation*}
(\Egs{\sr}(V)-\E{\sr}(\p_*))\,\|\psi_1\|^2
\le \sps{\psi_1}{h_{\sr}(V)\, \psi_1}-\const\, 
\sps{\psi_1}{\Shat\,\psi_1}\,,
\quad\hat{\p}:=-i\nabla_\x\,,
\end{equation*}
where $\const>0$ depends only on $\pmax_*\equiv\pmax_*(d)$, $d$, $g$, and $r$.
Applying \eqref{joice} 
yields
\begin{equation}\label{john}
(\Egs{\sr}(V)-\Egs{\sr}(0))\,\|\psi_1\|^2
\le \sps{\psi_1}{h_{\sr}(V)\, \psi_1}-\const\, 
\sps{\psi_1}{\Shat\,\psi_1}\,.
\end{equation}
Since $C_0^\infty(\R^3)$ is a form core for $h_\sr(V)$ the previous
bound actually
holds true, for every real-valued normalized
$\psi_1\in\form(h_\sr(V))$.

\smallskip

\noindent
(a): First, we prove the increase of binding energy. 
Under the conditions of Theorem~\ref{thm-binding}(a) we may choose
$\psi_1$ to be a real-valued, normalized ground state
eigenfunction of $h_{\sr}(V)$ corresponding to the
ground state energy $e_\sr(V)<0$. Then 
$\sps{\psi_1}{\Shat\,\psi_1}>0$
is a constant depending only on $V$ and
\eqref{john} yields
\begin{equation*}
\Egs{\sr}(V)-\Egs{\sr}(0)
\le e_\sr(V) -\const\,\sps{\psi_1}{\Shat\,\psi_1}\,.
\end{equation*}
(b): Next, we consider enhanced binding abilities.
Let the conditions of Theorem~\ref{thm-binding}(b)
be satisfied so that $\mu=1$ is the coupling constant threshold
for the family of potentials $V_\mu=V_+-\mu\,V_-$.
For $\lambda>1$, let $\psi_\lambda$ be the 
positive eigenvector of $h_{\sr}(V_\lambda)$ 
corresponding to the ground state energy $e_\lambda:=e_\sr(V_\lambda)<0$. 
We require that the eigenvector 
of the Birman-Schwinger operator $\Bir{\sr}{e_\lambda}{V_\lambda}$
(defined in \eqref{def-Bir} below)
corresponding to $\psi_\lambda$ is normalized.
(Compare Lemma~\ref{LemBirSchwing}, where we recall the
appropriate Birman-Schwinger principle.)
For $0<\mu\le1<\lambda$, we then infer from \eqref{john} 
with $\psi_1:=\psi_\lambda$ that
\begin{align}\nonumber
(\Egs{\sr}(V_\mu)-\Egs{\sr}(0))\|\psi_\lambda\|^2&
\le \sps{\psi_\lambda}{h_{\sr}(V_\mu)\,\psi_\lambda}
-\const\,\sps{\psi_\lambda}{\Shat\,\psi_\lambda}
\\\label{john2}
&\le e_\lambda\|\psi_\lambda\|^2 
-(\mu-\lambda)\sps{\psi_\lambda}{V_-\psi_\lambda}
-\const\,\sps{\psi_\lambda}{\Shat\,\psi_\lambda}\,.
\end{align}
Now, by Theorem~\ref{zero-resonance}
there exist $\lambda_j\downarrow1$ such that
$\{\Shat^{\nf{1}{2}}\,\psi_{\lambda_j}\}_j$
converges to some non-zero limit and, hence, 
$
\|\Shat^{\nf{1}{2}}\,\psi_{\lambda_j}\|\ge\alpha
$, for some $\alpha>0$ and large $j$.
Furthermore, the normalization condition imposed on $\psi_\lambda$
precisely says that
$\|\lambda\,\sps{\psi_\lambda}{V_-\,\psi_\lambda}\|=1$, $\lambda>1$;
see Lemma~\ref{LemBirSchwing}. 
Taking these remarks into account we deduce from \eqref{john2} that 
\begin{equation}\label{john3}
(\Egs{\sr}(V_\mu)-\Egs{\sr}(0))\,\|\psi_{\lambda_j}\|^2 
\le (\lambda_j-\mu)/ \lambda_j-\const\,\alpha^2,
\end{equation}
for sufficiently large $j$. 
Now, we conclude as in the NR case.
\end{proof}

%%%%%%%%%%%%%%%%%%%%%%%%%%%%%%%%%%%%%%%%%%%%%%%%%%%%%%%%%%%%%%%%%%%%%%%%%%
%%%%%%%%%%%%%%%%%%%%%%%%%%%%%%%%%%%%%%%%%%%%%%%%%%%%%%%%%%%%%%%%%%%%%%%%%%
%%%%%%%%%%%%%%%%%%%%%%%%%%%%%%%%%%%%%%%%%%%%%%%%%%%%%%%%%%%%%%%%%%%%%%%%%%
%
\appendix
%
%%%%%%%%%%%%%%%%%%%%%%%%%%%%%%%%%%%%%%%%%%%%%%%%%%%%%%%%%%%%%%%%%%%%%%%%%%
%%%%%%%%%%%%%%%%%%%%%%%%%%%%%%%%%%%%%%%%%%%%%%%%%%%%%%%%%%%%%%%%%%%%%%%%%%
%%%%%%%%%%%%%%%%%%%%%%%%%%%%%%%%%%%%%%%%%%%%%%%%%%%%%%%%%%%%%%%%%%%%%%%%%%

\section{Birman-Schwinger operators and zero-resonances}
\label{app-BS}

\noindent
In this section we consider the electronic one-particle Hamiltonians 
$$
h_\sharp:=h_\sharp(V)=\That{\sharp}+V\,,\qquad\sharp\in\{\nr,\sr\}\,,
$$
acting in the Hilbert space $L^2(\R^3)$, with
\begin{equation}\label{AppC-eq1}
 \That{\nr}=\tfrac{1}{2}\,\hat{\p}^2,\qquad\That{\sr}=(\hat{\p}^2+1)^{\nf{1}{2}}-1\,,
\qquad\hat{\p}=-\imath\nabla_{\x}\,,
\end{equation}
and for a certain class of short range potentials $V$.
We shall first recall the Birman-Schwinger principle for energies $e<0$
(Subsection~\ref{ssec-BS-nonsing}). After that we
discuss the Birman-Schwinger kernels 
and the existence and some
properties of zero-resonances 
in the singular limit $e\uparrow0$
(Subsection~\ref{ssec-BS-sing}).

%%%%%%%%%%%%%%%%%%%%%%%%%%%%%%%%%%%%%%%%%%%%%%%%%%%%%%%%%%%%%%%%%%%%%%%%%%%%%%%%
 
\subsection{Non-singular Birman-Schwinger kernels}\label{ssec-BS-nonsing}

\noindent
Let $\sharp\in\{\nr,\sr\}$
and let $V_+\ge0$ and  $V_-\ge0$ be the positive and 
negative parts of $V=V_+-V_-$, respectively. Assume that $V_\pm\in L^1_{\textrm{loc}}(\R^3)$
and that $V_-$ is $\That{\sharp}$-form-bounded with 
relative form bound $a<1$. Then
$\That{\sharp}+V_+$ and $\That{\sharp}+V$ define semi-bounded, closed
quadratic forms.
Let $h_\sharp^+$ and $h_\sharp$ denote the self-adjoint
operators representing these forms, respectively.
By the KLMN theorem $\form(h_\sharp^+)=\form(h_\sharp)\subset\dom(V_-^{\nf{1}{2}})$.
Then we may define a Birman-Schwinger operator, 
$\Bir{\sharp}{e}{V}$, for every $e<0$, 
by the formulas
\begin{equation}\label{def-Bir}
\Bir{\sharp}{e}{V}:= Y_\sharp(e)Y_\sharp(e)^*,\quad 
Y_\sharp(e):= V_-^{\nf{1}{2}}\,(h_\sharp^+-e)^{-\nf{1}{2}}.
\end{equation}
In fact, 
$Y_\sharp(e)$ is well-defined on $L^2(\RR^3)$ and bounded by the closed graph theorem
and $\Bir{\sharp}{e}{V}\in\cB(L^2(\R^3))$ is self-adjoint.
In the next lemma we compare the eigenspaces
\begin{align*}
\sB_\sharp(e)&:=\{ \psi\in L^2(\R^3)\,:\, \Bir{\sharp}{e}{V}\,\psi=\psi\}\,,
\\
\sF_\sharp(e)&:=\{ \phi\in\dom(h_\sharp)\,:\, h_\sharp\, \phi= e\,\phi\}\,.
\end{align*}

\begin{lemma}\label{LemBirSchwing}
Let $\sharp\in\{\nr,\sr\}$. Under the above assumptions on $V$
and, for every $e<0$,
there is a linear bijection 
$b\equiv b_\sharp(e):\sB_\sharp(e)\to\sF_\sharp(e)$
satisfying $\|b\,\psi\|\le |e|^{-\nf{1}{2}}\,\|\psi\|$, for all
$\psi\in\sB_\sharp(e)$. It is given by
\begin{align}\label{for-b}
b\,\psi&:=(h_\sharp^+-e)^{-\nf{1}{2}}\,Y_\sharp(e)^*\,\psi\,,\quad\psi\in\sB_\sharp(e)\,,
\\\label{for-binv}
b^{-1}\phi&=V_-^{\nf{1}{2}}\phi\,,\quad\phi\in\sF_\sharp(e)\,.
\end{align}
\end{lemma}
%
%%%%%%%%%%%%%%%%%%%%%%%%%%%%%%%%%%%%%%%%%%%%%%%%%%%
%
\begin{proof}
Assume that $e<0$ is an eigenvalue of $h_\sharp$ and
$\phi$ is a corresponding normalized eigenfunction.
For $\eta\in \form(h_\sharp)$, we then get
\begin{equation}\label{harry}
\sps{(h_\sharp^+-e)^{\nf{1}{2}}\,\eta}{(h_\sharp^+-e)^{\nf{1}{2}}\,\phi}
=\sps{V_-^{\nf{1}{2}}\,\eta}{V_-^{\nf{1}{2}}\,\phi}.
\end{equation}
Now, let $\eta:=(h_\sharp^+-e)^{-\nf{1}{2}}\,Y_\sharp(e)^*\,\eta'$ with 
some arbitrary $\eta'\in \dom(V_-^{\nf{1}{2}})$ 
and set $\psi:= V_-^{\nf{1}{2}}\,\phi\in L^{2}(\R^3)$. 
The condition $\eta'\in \dom(V_-^{\nf{1}{2}})$ 
ensures that $Y_\sharp(e)^*\,\eta'=(h_\sharp^+-e)^{-\nf{1}{2}}\,V_-^{\nf{1}{2}}\eta'$,
whence the LHS of \eqref{harry} becomes
$\sps{V_-^{\nf{1}{2}}\eta'}{\phi}=\sps{\eta'}{\psi}$.
We thus obtain
$\sps{\eta'}{\psi}=\sps{Y_\sharp(e)\,Y_\sharp(e)^*\,\eta'}{\psi}$ 
and conclude that
$\Bir{\sharp}{e}{V}\,\psi=\psi$. Suppose $\psi=0$. Then 
$0>e=\sps{\phi}{h_\sharp^+\,\phi}\ge 0$ by \eqref{harry} with $\eta=\phi$; a contradiction!

Conversely, assume that $\psi\in\sB_\sharp(e)$, $\psi\not=0$. Defining
$\phi:=b\,\psi\in \form(h_\sharp)$ as in \eqref{for-b} 
we obtain, for all $\eta\in \form(h_\sharp)$,
\begin{align*}
\sps{\eta}{(h_\sharp-e) \phi}
&= \sps{(h_\sharp^+-e)^{\nf{1}{2}}\,\eta}{Y_\sharp(e)^*\,\psi}
-\sps{V_-^{\nf{1}{2}}\,\eta}{V_-^{\nf{1}{2}}\,\phi}
\\
&= \sps{V_-^{\nf{1}{2}}\,\eta}{\psi} - \sps{V_-^{\nf{1}{2}}\,\eta}{V_-^{\nf{1}{2}}\,\phi}
\\
&=\sps{V_-^{\nf{1}{2}}\,\eta}{\psi} - \sps{V_-^{\nf{1}{2}}\,\eta}{Y_\sharp(e)\,Y_\sharp(e)^*\,\psi}=0\,.
\end{align*}
We deduce that $\phi\in \dom(h_\sharp)$ and $h_\sharp\,\phi=e\,\phi$.
Suppose $\phi=0$. Then $0=V_-^{\nf{1}{2}}\phi=\Bir{\sharp}{e}{V}\,\psi=\psi$,
which yields a contradiction.

Finally, we have
$\|Y_\sharp(e)^*\,\psi\|^2=\sps{\psi}{\Bir{\sharp}{e}\,\psi}=\|\psi\|^2$
and, hence,
$\|b\,\psi\|\le|e|^{-\nf{1}{2}}\|Y_\sharp(e)^*\,\psi\|=\|\psi\|$, 
for all $\psi\in\sB_\sharp(e)$.
\end{proof}

%%%%%%%%%%%%%%%%%%%%%%%%%%%%%%%%%%%%%%%%%%%%%%%%%%%%%%%%%%%%%%%%%%%%

\subsection{The singular Birman-Schwinger kernel}\label{ssec-BS-sing}

\noindent
In this subsection we study the limit $\tslim{e\uparrow 0}\Bir{\sharp}{e}{V}$.
To this end we restrict ourselves to potentials
$V_\pm\in L^1_\loc(\R^3)$ with the negative part satisfying
$V_-\in L^{\nf{3}{2}}(\R^3)$ in the non-relativistic case and
$V_-\in L^{\nf{3}{2}}\cap L^{3}(\R^3)$
in the semi-relativistic case.

Let $\gamma\in\{\nf{1}{2},1\}$.
Since $ |\cdot|^{-\gamma}\in L_w^{\nf{3}{\gamma}}(\R^3)$,
we know that the closure of the densely defined operator
$|\hat{\p}|^{-{\gamma}}\,V_-^{\nf{1}{2}}$ 
is compact, if $V_-^{\nf{1}{2}}\in L^{\nf{3}{\gamma}}(\R^3)$; 
see \cite{Cwikel1977}.
In particular, the closure of 
$Z_\nr:=\That{\nr}^{-\nf{1}{2}}\,V_-^{\nf{1}{2}}$
is compact.
Moreover,
$\That{\sr}^{-\nf{1}{2}}=f_<(\hat{\p})\,|\hat{\p}|^{-1}+f_>(\hat{\p})\,|\hat{\p}|^{-\nf{1}{2}}$
where $f_<$ is a bounded function supported in $\{|\p|\le1\}$
and $f_>$ is a bounded function supported in $\{|\p|\ge1\}$.
Therefore, the closure of 
$Z_\sr:=\That{\sr}^{-\nf{1}{2}}\,V_-^{\nf{1}{2}}$
is compact, too.
In particular, it follows that $V_-^{\nf{1}{2}}$ is relatively
compact with respect to $\That{\sharp}^{\nf{1}{2}}$ and, 
consequently,
$V_-$ is infinitesimally form-bounded with respect to $\That{\sharp}$.
(In the NR case the latter assertion also follows from the fact that
$V_-$ is a Rollnik potential.)
Therefore, the conclusions of Lemma~\ref{LemBirSchwing} are applicable
in what follows.

On account of $\That{\sharp}\le h_\sharp^+$ and the operator monotonicity
of the inversion we have
$\|(h_\sharp^+-e)^{-\nf{1}{2}}\,\That{\sharp}^{\nf{1}{2}}\psi\|\le\|\psi\|$,
for all $\psi\in\dom(\That{\sharp}^{\nf{1}{2}})$ and $e<0$.
Using the monotone convergence theorem in a spectral representation of
$h_\sharp^+$ we infer that
$\Ran(\That{\sharp}^{\nf{1}{2}})\subset\dom((h^+_\sharp)^{-\nf{1}{2}})$
and $\|(h^+_\sharp)^{-\nf{1}{2}}\,\That{\sharp}^{\nf{1}{2}}\|\le1$.
Therefore, the densely defined operators
$W_\sharp:=(h^+_\sharp)^{-\nf{1}{2}}\,\That{\sharp}^{\nf{1}{2}}$
and $X_\sharp:=(h^+_\sharp)^{-\nf{1}{2}}V_-^{\nf{1}{2}}$
have bounded extensions to the whole Hilbert space
and
$\ol{X}_\sharp=\ol{W}_\sharp\,\ol{Z}_\sharp$ is compact.
Furthermore, for $e<0$, it is straightforward to verify that
$Y_\sharp(e)^*=(h_\sharp^+)^{\nf{1}{2}}(h_\sharp^+-e)^{-\nf{1}{2}}\,\ol{X}_\sharp$.
Hence, $Y_\sharp(e)^*$ and
$Y_\sharp(e)=Y_\sharp(e)^{**}=X_\sharp^*(h_\sharp^+)^{\nf{1}{2}}(h_\sharp^+-e)^{-\nf{1}{2}}$ 
are compact and, 
on account of the spectral calculus and $\Ker(h_\sharp^+)=\{0\}$,
their
strong limits,
$$
Y_\sharp(0):=\slim{e\uparrow0}\,Y_\sharp(e)=X_\sharp^*\,,\qquad
\slim{e\uparrow0}\,Y_\sharp(e)^*=\ol{X}_\sharp=Y_\sharp(0)^*\,,
$$
exist and are compact. By virtue of the uniform boundedness principle
we may now define the singular Birman-Schwinger operator
\begin{equation}\label{MonLim}
\Bir{\sharp}{0}{V}
:= Y_\sharp(0)\,Y_\sharp(0)^*= \slim{e\uparrow 0}\,\Bir{\sharp}{e}{V}\,.
\end{equation}
The next theorem generalizes some results of 
\cite{SorStock} to a broader class of potentials.

\begin{theorem}\label{zero-resonance}
Let $\sharp\in\{\nr,\sr\}$, $0\le V_+\in L_\loc^1(\RR^3)$, and 
$0\le V_-\in L^{\nf{3}{2}}(\RR^3)$. If $\sharp=\sr$, then assume in
addition that $V_-\in L^3(\RR^3)$.
Set $V_\lambda:=V_+-\lambda\,V_-$ and assume that, for every
$\lambda>1$, there is an eigenvalue $e_\lambda<0$ of
$h_\sharp(V_\lambda)$ such that $e_\lambda\to0$, $\lambda\downarrow1$.
Let $\phi_\lambda$ be a corresponding eigenfunction 
such that the eigenfunction
$\psi_\lambda=(\lambda\,V_-)^{\nf{1}{2}}\,\phi_\lambda$
of the Birman-Schwinger operator has norm $1$.
Then there is a sequence $\{\lambda_j\}_j$, $\lambda_j\downarrow1$,
$j\uparrow\infty$,
such that the limits
$\psi:=\lim_j\psi_{\lambda_j}$,
$\rho:=\lim_j(h_\sharp^+)^{\nf{1}{2}}\phi_{\lambda_j}$,
and $\tilde{\rho}:=\lim_j\That{\sharp}^{\nf{1}{2}}\phi_{\lambda_j}$
exist and 
$$
\psi=\Bir{\sharp}{0}{V}\,\psi\,,\quad
\rho=Y_\sharp(0)^*Y_\sharp(0)\,\rho\not=0\,,\quad\tilde{\rho}\not=0\,.
$$
Moreover, if $\sharp=\nr$, then the limit
$\phi:=\lim_j\phi_{\lambda_j}$ exists in $L^6(\R^3)$.
If $\sharp=\sr$, then
$\phi_<:=\lim_j\id_{|\hat{\p}|\le1}\,\phi_{\lambda_j}$
exists in $L^6(\R^3)$,
$\phi_>:=\lim_j\id_{|\hat{\p}|>1}\,\phi_{\lambda_j}$
exists in $L^3(\R^3)$, and we set $\phi:=\phi_<+\phi_>$.
In both cases $\phi$
is a zero resonance, i.e. a weak solution of 
$h_\sharp\,\phi\equiv h_\sharp(V_+-V_-)\,\phi=0$ in the sense that
\begin{equation}\label{weak-sol}
\int_{\RR^3}{\phi}\,(\That{\sharp}\,\eta+V_+\,\eta-V_-\eta)=0\,,
\qquad\eta\in C_0^\infty(\RR^3)\,.
\end{equation}
\end{theorem}

\begin{proof}
By Lemma~\ref{LemBirSchwing}, 
$\|\phi_\lambda\|\le|e_\lambda|^{-\nf{1}{2}}\,\|\psi_\lambda\|
=|e_\lambda|^{-\nf{1}{2}}$
which implies,
for all $\eta\in\form(h_\sharp^+)\subset\dom(V_-^{\nf{1}{2}})$,
\begin{align}\label{cassandra}
\sps{h_{\sharp}\,\eta}{\phi_{\lambda}}
&=e_{\lambda}\,\sps{\eta}{\phi_{\lambda}}
+(\lambda-1)\,\sps{V_-^{\nf{1}{2}}\,\eta}{\psi_{\lambda}}/\lambda^{\nf{1}{2}}
\xrightarrow{\;\lambda\downarrow1\;}0\,,
\\\nonumber
\sps{h_{\sharp}^+\,\phi_{\lambda}}{\phi_{\lambda}}
&=e_{\lambda}\,\|\phi_{\lambda}\|^2+\|\psi_\lambda\|^2/\lambda\le1\,,
\quad\lambda>1\,.
\end{align}
Therefore, we find a sequence, $\{\lambda_j\}_j$, 
$\lambda_j\downarrow1$, such that the
weak limits $\psi:=\twlim{j}\psi_{\lambda_j}$ 
and $\rho:=\twlim{j}(h_\sharp^+)^{\nf{1}{2}}\,\phi_{\lambda_j}$
exist. 
We have $\psi=Y_\sharp(0)\,\rho$ because
\begin{align*}%\label{viona}
\sps{\eta}{\psi}
&=\lim_j\sps{\eta}{\psi_{\lambda_j}}
=\lim_j\sps{V_-^{\nf{1}{2}}\eta}{\phi_{\lambda_j}}
=\sps{X_\sharp\,\eta}{\rho}\,,\quad \eta\in\dom(X_\sharp)\,.
\end{align*}
In particular,
$\rho=0$ implies $\psi=0$.
Furthermore,
since $Y_\sharp(0)$ is compact we know that
$\{Y_\sharp(0)\,(h_\sharp^+)^{\nf{1}{2}}\phi_{\lambda_j}\}_j$
contains a strongly convergent subsequence.
As it converges weakly to $Y_\sharp(0)\,\rho=\psi$
we may assume that $Y_\sharp(0)\,(h_\sharp^+)^{\nf{1}{2}}\phi_{\lambda_j}\to\psi$
strongly, after passing to a subsequence, if necessary.
By the Birman-Schwinger principle of Lemma~\ref{LemBirSchwing}
we know, however, that
$$
\phi_{\lambda_j}=(h_\sharp^+-e_{\lambda_j})^{-\nf{1}{2}}\,Y_\sharp(e_{\lambda_j})^*\,\psi_{\lambda_j}
=(h_\sharp^+)^{\nf{1}{2}}(h_\sharp^+-e_{\lambda_j})^{-1}\,Y_\sharp(0)^*\,\psi_{\lambda_j}\,.
$$
Since $Y_\sharp(0)^*$ is compact we may further assume
that $Y_\sharp(0)^*\,\psi_{\lambda_j}\to Y_\sharp(0)^*\,\psi$ strongly,
whence
$$
(h_\sharp^+)^{\nf{1}{2}}\phi_{\lambda_j}
=h_\sharp^+\,(h_\sharp^+-e_{\lambda_j})^{-1}\,Y_\sharp(0)^*\,\psi_{\lambda_j}\to
Y_\sharp(0)^*\,\psi\,.
$$
Hence, 
$\rho=\lim_j(h_\sharp^+)^{\nf{1}{2}}\phi_{\lambda_j}=Y_\sharp(0)^*\,\psi=Y_\sharp(0)^*Y_\sharp(0)\,\rho$ 
converges strongly and so does
$\psi=\lim_jY_\sharp(0)\,(h_\sharp^+)^{\nf{1}{2}}\,\phi_{\lambda_j}=Y_\sharp(0)\,Y_\sharp(0)^*\,\psi$.
Since also $Y_\sharp(0)\,(h_\sharp^+)^{\nf{1}{2}}\eta=
X_\sharp^*(h_\sharp^+)^{\nf{1}{2}}\eta=V_-^{\nf{1}{2}}\eta$,
for 
$\eta\in\form(h_\sharp^+)\subset\dom(V_-^{\nf{1}{2}})$,
we have
$Y_\sharp(0)\,(h_\sharp^+)^{\nf{1}{2}}\phi_{\lambda_j}=\lambda_j^{-\nf{1}{2}}\,\psi_{\lambda_j}$.
It follows that $\psi_{\lambda_j}\to\psi$ strongly, $\|\psi\|=1$, 
and $\rho\not=0$.
Furthermore, it follows that
$\That{\sharp}^{\nf{1}{2}}\,\phi_{\lambda_j}=W_\sharp^*\,(h_\sharp^+)^{\nf{1}{2}}\,\phi_{\lambda_j}\to
W_\sharp^*\rho$ strongly, where $W_\sharp^*\rho\not=0$
since $W_\sharp^*$ is bijective.
On account of Sobolev's inequalities for (half-)derivatives
this implies the existence of the limits $\phi$, $\phi_<$, $\phi_>$
as in the statement. Since, for every $\eta\in C_0^\infty(\R^3)$,
we have $V_\pm\eta\in L^{6/5}\cap L^{3/2}(\R^3)$
we finally see that \eqref{weak-sol} follows from \eqref{cassandra}.
\end{proof}

%%%%%%%%%%%%%%%%%%%%%%%%%%%%%%%%%%%%%%%%%%%%%%%%%%%%%%%%%%%%%%
%%%%%%%%%%%%%%%%%%%%%%%%%%%%%%%%%%%%%%%%%%%%%%%%%%%%%%%%%%%%%%
%%%%%%%%%%%%%%%%%%%%%%%%%%%%%%%%%%%%%%%%%%%%%%%%%%%%%%%%%%%%%%

\bigskip

{\footnotesize

\noindent
{\sc Martin K\"onenberg}:
Fakult\"at f\"ur Physik,
Universit\"at Wien,
Boltzmanngasse 5,
1090 Vienna, Austria.
{\tt martin.koenenberg@univie.ac.at}
 
\smallskip

\noindent
{\sc Oliver Matte}: 
Mathematisches Institut,
Ludwig-Maximilians-Universit\"at,
Theresienstra{\ss}e 39,
80333 M\"unchen, Germany.\\
{\em Present address:}
Institut for Matematik,
{\AA}rhus Universitet,
Ny Munkegade 118,
DK-8000 {\AA}rhus, Denmark.
{\tt matte@math.lmu.de}
}
\end{document}